\documentclass[format=acmsmall, review=false, screen=true]{acmart}
\usepackage{booktabs} % For formal tables
\acmJournal{tmis} % TEAC  TEAC  TSC TKDD CSUR TECS
\acmVolume{1}
\acmNumber{3}
\acmArticle{18}
\acmYear{2019}
\acmMonth{7}

\setcopyright{acmcopyright}
\acmDOI{0001.0001}

\newtheorem{strategy}{Strategy}    % gws
\usepackage{multirow}
\usepackage[noend]{algpseudocode}
\usepackage[ruled]{algorithm2e} % For algorithms

\SetAlFnt{\small}
\SetAlCapFnt{\small}
\SetAlCapNameFnt{\small}
\SetAlCapHSkip{0pt}
\IncMargin{-\parindent}

\usepackage{array}
\newcolumntype{P}[1]{>{\centering\arraybackslash}p{#1}}
\newcolumntype{M}[1]{>{\centering\arraybackslash}m{#1}}
\newcolumntype{R}[1]{>{\arraybackslash}m{#1}}

\definecolor{orange}{rgb}{1,0.5,0}
\definecolor{graynode}{RGB}{20,20,20}
\definecolor{crimsonred}{RGB}{220,20,60}
\definecolor{darkgraynode}{gray}{0.5}
\definecolor{lightgraynode}{gray}{0.8}

\usepackage{rotate}
\usepackage{capt-of}
\usepackage{tabulary}
\usepackage{setspace}
\usepackage{amssymb}
\usepackage{pifont}

\definecolor{gray}{RGB}{20,20,20}
\definecolor{gray}{RGB}{0.7,0.7,0.7}
\definecolor{greencm}{RGB}{0,153,0}

\definecolor{plotblue}{RGB}	{30,144,255}
\definecolor{plotgreen}{RGB}	{50,205,50}
\definecolor{plotred}{RGB}	{220,20,60}

\definecolor{myyellow}{RGB}{255,255,204}
\definecolor{myred}{RGB}{255,204,204}
\definecolor{myblue}{RGB}{0,200,255}
\definecolor{mygreen}{RGB}{80,220,80}

\newcommand*\hrulefillvar[1][0.4pt]{\leavevmode\leaders\hrule height#1\hfill\kern0pt}

\usepackage{ragged2e}

\usepackage{comment}
\usepackage{tabularx}

\usepackage{subfigure}
\usepackage{enumitem}

\definecolor{thedarkblue}{RGB}{0,0,120} %104} % 180
\definecolor{mydarkblue}{rgb}{0,0.08,0.45} %ICML dark blue

\usepackage{hyperref}
\hypersetup{%
    colorlinks=true,
    linkcolor=mydarkblue,
    citecolor=mydarkblue,
    filecolor=mydarkblue,
    urlcolor=mydarkblue}
    
\usepackage{microtype}
\usepackage{graphicx}

\usepackage{balance}
\usepackage{tabularx}

\usepackage{booktabs}
\usepackage{tabularx}
\usepackage{comment}

\begin{document}

\title{Utility-Driven Mining of Trend Information for Intelligent System}

\author{Wensheng Gan}
\affiliation{%
		\institution{Harbin Institute of Technology (Shenzhen)}
		\city{Shenzhen}
		\country{China}
}
\email{wsgan001@gmail.com}

\author{Jerry Chun-Wei Lin}
\authornote{Corresponding author}
\affiliation{%
	\institution{Western Norway University of Applied Sciences (HVL)}
	\city{Bergen}
	\country{Norway}
}
\email{jerrylin@ieee.org}

\author{Han-Chieh Chao}
\affiliation{%
	\institution{National Dong Hwa University}
	\city{Hualien}
	\country{Taiwan}
}
\email{hcc@ndhu.edu.tw}

\author{Philippe Fournier-Viger}
\affiliation{%
	\institution{Harbin Institute of Technology (Shenzhen)}
	\city{Shenzhen}
	\country{China}	
}
\email{philfv@hitsz.edu.cn}

\author{Xuan Wang}
\affiliation{%
	\institution{Harbin Institute of Technology (Shenzhen)}
	\city{Shenzhen}
	\country{China}	
}
\email{wangxuan@cs.hitsz.edu.cn}

\author{Philip S. Yu}
\affiliation{%
	\institution{University of Illinois at Chicago}
	\city{Chicago}
	\country{USA}
}
\email{psyu@uic.edu}

\renewcommand\shortauthors{W. Gan et al.}

\begin{abstract}

Useful knowledge, embedded in a database, is likely to change over time. Identifying recent changes in temporal databases can provide valuable up-to-date information to decision-makers. Nevertheless, techniques for mining high-utility patterns (HUPs) seldom consider  recency as a criterion to discover patterns. Thus, the traditional utility mining framework is inadequate for obtaining up-to-date insights about real world data. In this paper, we address this issue by introducing a novel framework, named utility-driven mining of Recent/trend high-Utility Patterns (RUP) in temporal databases for intelligent systems, based on user-specified minimum recency and minimum utility thresholds. The utility-driven RUP algorithm is based on novel global and conditional downward closure properties, and a recency-utility tree. Moreover, it adopts a vertical compact recency-utility list structure to store the information required by the mining process.  The developed RUP algorithm recursively discovers recent HUPs. It is also fast and consumes a small amount of memory due to its pattern discovery approach that does not generate candidates. Two improved versions of the algorithm with additional pruning strategies are also designed to speed up the discovery of patterns by reducing the search space. Results of a substantial experimental evaluation show that the proposed algorithm can efficiently identify all recent high-utility patterns in large-scale databases, and  that the improved  algorithm performs best.

\end{abstract}

\begin{CCSXML}
<ccs2012>
 <concept>
<concept_id>10002951.10003227.10003351</concept_id>
<concept_desc>Information systems~Data mining</concept_desc>
<concept_significance>500</concept_significance>
</concept>

</ccs2012>
\end{CCSXML}

\ccsdesc[500]{Information Systems~Data mining}

\ccsdesc[500]{Applied Computing~Business intelligence} 

\keywords{Economic behavior, utility theory, temporal database, high-utility patterns, recency.}

\maketitle

%%%%%%%%%%%%%%%%%%%%%%%%%%%%%%%%%%%%%%%%%%%%%%%%%%%%%%%%%%%%%%%%%%

\section{Introduction}
Knowledge Discovery in Database (KDD) \cite{agrawal1993database,chen1996data,gan2017data,gan2019survey} aims at extracting meaningful and useful information from massive amounts of data \cite{agrawal1994fast,han2004mining,gan2019survey}. Frequent pattern mining (FPM) and association rule mining (ARM) \cite{agrawal1993database,agrawal1994fast,han2004mining} are some of the most important and fundamental KDD techniques \cite{gan2017data}, which satisfy the requirements of applications in numerous domains. Mining frequent patterns consists of discovering all the frequent patterns in a given database based on a user-specified minimum support threshold. Different from the tasks of FPM and ARM, high-utility pattern mining (HUPM) \cite{chan2003mining,yao2004foundational,lin2016efficient,liu2012mining,liu2005two} takes into account both purchase quantities and profit values of items to measure how ``useful''  items and patterns are. A pattern is considered to be a high-utility pattern (HUP) if its utility  in a database is no less than a user-specified minimum utility count. The ``utility'' measure can be generally viewed as a measure of the importance of patterns to the user, for example, in terms of cost, risk, or unit profit. Chan \textit{et al.} \cite{chan2003mining} first established the framework of HUPM. Yao \textit{et al.} \cite{yao2004foundational} then defined a unified framework for mining high-utility patterns (HUPs). HUPM is useful as it can identify patterns that may be infrequent in transactions, but that are profitable, and thus highly valuable to retailers or managers. HUPM is a key data analysis task that has been widely applied to discover useful knowledge in different types of massive data. In order to improve the mining efficiency of utility mining, many HUPM approaches have been proposed, such as Two-Phase \cite{liu2005two}, IHUP \cite{ahmed2009efficient}, UP-growth+ \cite{tseng2013efficient}, HUI-Miner \cite{liu2012mining}, and EFIM \cite{zida2017efim}. In addition, many interesting issues related to the effectiveness of utility-driven pattern mining have been extensively studied, as summarized in \cite{gan2018survey}.

However, an important drawback of previous studies is that they utilize the minimum utility threshold as sole constraint for discovering HUPs, and ignores the temporal aspect of databases. In general, knowledge found in a temporal database   changes as time goes by. Extracting up-to-date knowledge, especially from temporal data, can provide valuable information for many real-world applications, such as decision making \cite{chang2003finding,hong2009effective} and event detection \cite{batal2012mining}. Although HUPs can reveal  information  that is often more useful than frequent patterns, HUPM does not assess how recent the discovered patterns are. As a result, the discovered HUPs may be irrelevant or even misleading if they are out-of-date in many real-world use-cases. For example, if a pattern only appears in the far past, it may possibly be invalid at present, and thus be useless for decision making. Hence, it is unfair to measure the interestingness of a pattern without considering how recent the pattern is. For monitoring and event detection problems \cite{batal2012mining}, the temporal patterns in Electronic Health Record (EHR)  are potentially more useful for monitoring  a specific patient's condition. For market basket analysis, obtaining  information about recent or current sale trends is crucial, and much more important than gaining information about previous sale trends. Managers and retailers may use up-to-date patterns to take strategic business decisions, while out-of-date patterns may be useless or even misleading for this purpose. For example, consider that transactions in  a retail store contain the following five items: \{\textit{bread}\}, \{\textit{ice cream}\}, \{\textit{mitten}\}, \{\textit{sweater}\} and \{\textit{sunscreen}\}. 
Items such as \{\textit{ice cream}\} and \{\textit{sunscreen}\} are generally best selling during the summer, while  items such as \{\textit{mitten}\} and \{\textit{sweater}\} have stronger sales during the winter. 
Hence, for the purpose of market analysis, it is more valuable  to discover up-to-date and interesting patterns that have sold well recently indicating current hot sale trends,  than  patterns that have mostly appeared in the past. Because high-utility patterns  mostly sold  in the  past are generally irrelevant, traditional HUPM algorithms are  inappropriate for obtaining a set of up-to-date patterns representing recent trends.

Up to now, few studies have addressed the problem of mining recent HUPs by considering time sensitive constraint. To address the lack of consideration for the temporal aspect of real-world data in HUPM, the UDHUP-apriori and UDHUP-list algorithms \cite{lin2015efficient} have been proposed. These algorithms find up-to-date patterns having high utilities in databases. To the best of our knowledge,  UDHUP-apriori and UDHUP-list \cite{lin2015efficient} are the first algorithms that address this important issue of mining up-to-date high-utility patterns (UDHUPs). An UDHUP is a pattern which may not be a HUP when considering an entire transaction database, but that has been recently highly profitable. In other words, UDHUP mining algorithms  focus on finding  recently profitable patterns rather than patterns that have been globally profitable. Discovering UDHUPs is a difficult problem as these algorithms may face a ``combinatorial explosion'' of patterns since the number of recent patterns may be very large, especially for databases having many long transactions or when the minimum utility threshold is set low. 

To resolve the problem of utility-driven efficiently and effectively mining of trend information for intelligent system, a tree-based algorithm named mining Recent high-Utility Patterns in temporal databases (RUP) for intelligent systems is presented in this paper. The major contributions of this study are summarized as follows:

\begin{itemize}
	\item A novel utility-driven mining approach named mining Recent high-Utility Patterns in temporal databases (RUP) is proposed to reveal useful and meaningful recent high-utility patterns (RHUPs) by considering temporal factor. These patterns are more useful for real-life applications than discovering HUPs.
	
	\item The developed RUP approach spans a Set-enumeration tree named Recency-Utility tree (RU-tree). Using the designed recency-utility list (RU-list) structure, RUP only performs two database scans to find RHUPs and can avoid generating a huge number of candidate patterns.
	
	\item Two properties, named the global downward closure (GDC) property and the conditional downward closure (CDC) property,  ensures global and partial anti-monotonicity for mining RHUPs in the RU-tree. Moreover, additional effective pruning strategies are integrated in improved versions of RUP. RUP algorithm can greatly reduce the search space and thus speed up the discovery of RHUPs.	
	
	\item Substantial experiments  show that the proposed RUP algorithm and its improved versions can efficiently identify recent high-utility patterns appearing in recent transactions, and that the developed algorithm outperform the conventional algorithm (e.g., FHM) for HUPM. Moreover, it is shown that the improved variations of RUP outperforms the baseline.
\end{itemize}

Note that this paper is extended by a preliminary version \cite{gan2016mining}. The rest of this paper is organized as follows. Related work is  reviewed in Section 2. Preliminaries related to recent high-utility pattern and the problem statement are described in Section 3. Section 4 presents the proposed RUP algorithm. An experimental evaluation of its performance is provided in Section 5. Lastly, a conclusion is drawn in Section 6.

\section{Related Work}

\subsection{Efficiency Issue for Utility Mining}

Based on the numerous studies, it has been shown that pattern mining is useful in many real-world applications, such as decision making in intelligent systems (i.e., stock investment) \cite{lin2015discovering}, condition monitoring \cite{batal2012mining}, event detection \cite{batal2012mining}, pattern-based prediction and classification \cite{cheng2007discriminative}. 
HUPM is different from frequent pattern mining (FPM) since it takes purchase quantities and unit profits of items into account to determine the importance of  patterns, rather than simply considering their occurrence frequencies. Several studies have been carried out on HUPM. Chan \textit{et al.} \cite{chan2003mining} presented a framework to mine the top-\textit{k} closed utility patterns based on business objectives. Yao \textit{et al.} \cite{yao2004foundational} defined utility mining as the problem of finding profitable patterns while considering both the purchase quantities of items in transactions (internal utilities) and their unit profits (external utilities). Liu \textit{et al.} \cite{liu2005two} then presented a two phase algorithm,  adopting the transaction-weighted downward closure (\textit{TWDC}) property to efficiently discover HUPs. This approach has been named the transaction-weighted utilization (\textit{TWU}) model. Tseng \textit{et al.} \cite{tseng2010up} then proposed the  UP-growth+  algorithm to mine HUPs using an UP-tree structure. Different from previous approaches,  a novel utility-list-based algorithm named HUI-Miner \cite{liu2012mining} was designed to efficiently discover HUPs without generating candidates. This  algorithm greatly reduces the time required for discovering HUPs by not generating candidates. Enhanced versions of HUI-Miner, called FHM \cite{fournier2014fhm} and HUP-Miner \cite{krishnamoorthy2015pruning}, were respectively proposed, integrating additional effective pruning strategies. Experiments have shown that both FHM \cite{fournier2014fhm} and HUP-Miner \cite{krishnamoorthy2015pruning} significantly outperform the existing  HUPM algorithms for most static databases. Recently, several other novel approaches called EFIM \cite{zida2017efim} and HMiner \cite{krishnamoorthy2017hminer} were designed for mining HUPs. They were the state-of-the-art algorithms and shown to outperform previous methods. Consider the sequential data, the topic of high-utility sequential pattern  mining also has been studied, such as USpan \cite{yin2012uspan}, HUS-Span \cite{wang2016efficiently}, ProUM \cite{gan2019proum}, and HUSP-ULL \cite{gan2019fast}. Developing efficient high-utility pattern mining algorithms is an active research area, and more recent studies can be referred to review literature \cite{gan2018survey}.

\subsection{Effectiveness Issue for Utility Mining}

In addition to the above efficiency issue, many interesting effectiveness issues related to HUPM have been studied \cite{gan2018survey,lin2016efficient,lin2017fdhup}. A series of approaches for discovering HUPs in different types of dynamic data \cite{2gan2018survey,2lin2015mining,yun2017efficient} or uncertain data \cite{lin2016efficient} have been presented. Recently, a lattice-based method was presented for mining high-utility association rules \cite{mai2017lattice}. The FHN \cite{lin2016fhn} and GHUM \cite{krishnamoorthy2018efficiently} algorithms were proposed to discover HUPs when considering both positive and negative unit profit values. Lan \textit{et al.} \cite{lan2011discovery} first proposed a model for mining on-shelf HUPs, that is to discover profitable sets of products while considering their shelf-time. Algorithms have also been proposed to identify high utility occupancy patterns \cite{gan2018huopm} and find the top-$k$ HUPs \cite{tseng2016efficient,dam2017efficient} without setting minimum utility threshold.  Lee \textit{et al.} \cite{lee2013utility} proposed the task of mining utility-based association rule with applications to cross-selling. Besides, the topic of utility mining with different discount strategies in cross-selling has been studied \cite{lin2016fast}. Yun \textit{et al.} \cite{yun2018damped} developed several approaches to discover the high-utility patterns over data streams.  Another popular extension of the task of HUPM is to discover high average utility patterns, where the utility of a pattern is divided by the number of items that it contains \cite{hong2011effective,truong2018efficient,wu2018tub}. It has been argued that this measure provide a more fair measurement of the utilities of patterns. Because the set of high-utility patterns does not mean that they are correlated, Gan \textit{et al.} then introduced the interesting topic of correlated-based utility mining \cite{gan2017extracting}. Up to now, some studies of privacy-preserving utility mining also have been addressed, as reviewed in \cite{gan2018privacy}.

\subsection{Comparative Analysis with Related Work}

Information discovered by HUPM can be used to feed expert and intelligent systems and is more valuable than frequent patterns traditionally found in the field of pattern mining. However, most HUPM algorithms ignore  time factor, and thus do not consider how recent patterns are. In real-world situations, knowledge embedded in a database may change at any time. Hence,  discovered HUPs may be out-of-date and possibly invalid at present. Mining up-to-date information in temporal databases can provide valuable information to decision-makers. A new framework for mining up-to-date high-utility patterns (UDHUPs) \cite{lin2015efficient} was proposed to reveal useful and meaningful HUPs, by considering both the utility and the recency of patterns. The UDHUP algorithm extracts patterns which may not be globally profitable but that have been highly profitable in recent times. To the best of our knowledge, this is the first work to address the problem of mining up-to-date high-utility patterns. Discovering UDHUPs is a difficult problem as algorithms can face a ``combinatorial explosion'' of patterns since the number of recent patterns may be very large, especially for databases having many long transactions or when the minimum utility threshold is set low. Thus, it is a critical and challenging issue to discover recent HUPs more efficiently.

\section{Preliminaries and Problem Statement}
\subsection{Preliminaries}

Let \textit{I} = \{\textit{i}$_{1}$, \textit{i}$_{2}$, $\ldots$, \textit{i$_{m}$}\} be the set of the \textit{m} distinct items appearing in a temporal transactional database \textit{D} = \{\textit{T}$_{1}$, \textit{T}$_{2}$, $\ldots$, \textit{T$_{n}$}\}, such that each quantity transaction, denoted as \textit{T$_{q}$} $ \in $ \textit{D}, is a subset of \textit{I}, and has a unique transaction identifier (\textit{TID}) and a timestamp. Each item $i_{j} \in I$ is associated with a unit profit value $ pr(i_{j})$ representing the profit generated by the sale of a unit of item $i_{j}$. Unit profit values of items are stored in a profit-table \textit{ptable} = \{$ pr(i_{1}) $, $ pr(i_{2}) $, $\dots$, $ pr(i_{m}) $\}. Note that here the profit can also be defined as risk, interestingness, satisfaction, and other factor. An itemset/pattern \textit{X} $\in$ \textit{I} having \textit{k} distinct items \{$ i_{1} $, $ i_{2} $, $\dots$, $ i_{k} $\} is said to be  a \textit{k}-pattern. For an itemset/pattern $X$, the notation \textit{TIDs}$(X)$ denotes the \textit{TIDs} of all transactions in $D$ containing $X$.

Consider a time-varying e-commerce database\footnote{\url{https://recsys.acm.org/recsys15/challenge/}} that containing a set of temporal purchase behaviors. Then we illustrate the concept of RHUP using a simple e-commerce database presented in Table \ref{db:example}, which will be used as running example. This time-varying e-commerce database contains 10 purchase behaviors (transactions), sorted by purchase time. Moreover, assume that the corresponding \textit{ptable} of each product in Table \ref{db:example} is defined as \{\textit{pr}(\textit{a}): \$6, \textit{pr}(\textit{b}): \$1, \textit{pr}(\textit{c}): \$10, \textit{pr}(\textit{d}): \$7, \textit{pr}(\textit{e}): \$5\}.

\begin{table}[!htbp]
	\setlength{\abovecaptionskip}{0pt}
	\setlength{\belowcaptionskip}{0pt} 	
	\caption{An example transactional database}
	\label{db:example}
	\centering
		\begin{tabular}{|c|c|c|}
			\hline
			\textbf{TID} & \textbf{Timestamp}  & \textbf{Items with occurred quantities} \\ \hline \hline
			$ T_{1} $ &	2016/1/2 09:30   & 	\textit{a}:2, \textit{c}:1, \textit{d}:2 \\ \hline
			$ T_{2} $ &	2016/1/2 10:20   & 	\textit{b}:1, \textit{d}:2 \\ \hline
			$ T_{3} $ &	2016/1/3 19:35   &	\textit{b}:2, \textit{c}:1, \textit{e}:3 \\ \hline
			$ T_{4} $ &	2016/1/3 20:20   &	\textit{a}:3, \textit{c}:2 \\ \hline
			$ T_{5} $ &	2016/1/5 10:00   &	\textit{a}:1, \textit{b}:3, \textit{d}:4, \textit{e}:1 \\ \hline
			$ T_{6} $ &	2016/1/5 13:45   &	\textit{b}:4, \textit{e}:1 \\ \hline
			$ T_{7} $ &	2016/1/6 09:10   &	\textit{a}:3, \textit{c}:3, \textit{d}:2 \\ \hline
			$ T_{8} $ &	2016/1/6 09:44   &	\textit{b}:2, \textit{d}:3 \\ \hline
			$ T_{9} $ &	2016/1/6 16:10   &	\textit{c}:1, \textit{d}:2, \textit{e}:2 \\ \hline
			$ T_{10} $ & 2016/1/8 10:35 	&	\textit{a}:2, \textit{c}:2, \textit{d}:1 \\ \hline	
		\end{tabular}	
\end{table}

\begin{definition}
	\label{def_1}
	\rm The recency of a quantity transaction  $T_{q}$ is denoted as $ r(T_{q}) $ and defined as: $r(T_{q}) $ = $ (1-\delta)^{(T_{current}-T_{q})} $, where $ \delta $ is a user-specified time-decay factor ($ \delta \in$ (0,1]), $ T_{current} $ is the current timestamp which is equal to the number of transactions in $D$, and $T_q$ is the \textit{TID} of  transaction $T_q$. Then the recency of an itemset/pattern $X$ in a transaction $T_{q}$ is denoted as $ r(X, T_{q}) $ and defined as: $ r(X, T_{q}) $ = $r(T_{q}) $ = $(1-\delta)^{(T_{current}-T_{q})}$. The recency of an itemset/pattern $X$ in a database  $ D $ is denoted as $ r(X) $ and defined as: $ r(X)$ = $\sum _{X\subseteq T_{q} \wedge T_{q}\in D}r(X,T_{q})$.
\end{definition}

Note that the time-decay factor $ \delta $ is based on the prior knowledge. Thus, higher recency values are assigned to transactions having timestamps that are close to the most recent timestamp. For instance, assume that $ \delta $ is set to 0.1. The recency values of $T_1$ and $T_8$ are respectively calculated as $r(T_1)$ = $(1 - 0.1)^{(10-1)}$ = 0.3874, and $r(T_8)$ = $(1 - 0.1)^{(10-8)}$ = 0.8100. The recency values of all transactions in $D$ are $r(T_1)$ = 0.3874, $r(T_2)$ = 0.4305, $r(T_3)$ = 0.4783, $r(T_4)$ = 0.5314, $r(T_5)$ = 0.5905, $r(T_6)$ = 0.6561, $r(T_7)$ = 0.7290, $r(T_8)$ = 0.8100, $r(T_9)$ = 0.9000, and $ r(T_{10})$ = 1.000. Therefore, the recency of the 1-pattern $(b)$ in $T_2$ is calculated as $ r(b, T_2)$ = $r(T_2)$ = 0.4305, and the recency of the 2-pattern $ (bde) $ in $T_5$ is calculated as $ r(bde, T_5)$ = $r(T_5)$ = 0.5905. Consider patterns $ (be) $ and $ (bde) $ in $D$, their recency values are respectively calculated as $ r(be)$ = $r(be, T_3)$ + $r(be, T_5)$ + $r(be, T_6) $ = 0.4783 + 0.5905 + 0.6561 = 1.7249 and $ r(bde)$ = $r(bde, T_{5}) $ = 0.5905. To the best of our knowledge, incorporating the concept of recency in temporal utility mining  has not been previously explored in the utility mining literature, except for the concept to of up-to-date high-utility pattern (UDHUP) \cite{lin2015efficient}.

\begin{definition}
	\label{def_3}
	\rm The utility of an item $ i_{j} $ appearing in a quantity transaction $ T_{q} $ is denoted as $u(i_{j}, T_{q}) $, and is defined as: $ u(i_{j}, T_{q})$ = $ q(i_{j}, T_{q}) \times pr(i_{j}) $, in which $q(i_{j}, T_{q})$ is the quantity of $ i_{j} $ in $ T_{q} $, and  $pr(i_{j})$ is the unit price/utility of the item $ i_{j} $. The utility of an itemset/pattern \textit{X} in a transaction $ T_{q} $ is denoted as $ u(X, T_{q}) $ and defined as: $ u(X, T_{q})$ = $\sum _{i_{j}\in X\wedge X\subseteq T_{q}}u(i_{j}, T_{q})$. Let $ u(X) $ denote the utility of a pattern \textit{X} in a database \textit{D}, then it can be  defined as: $u(X)$ = $ \sum_{X\subseteq T_{q}\wedge T_{q} \in D} u(X, T_{q})$.
\end{definition}

For example, the utility of the item (\textit{c}) in transaction $ T_{1} $ is calculated as $ u(c, T_{1})$ = $q(c, T_{1})\times pr(c) $ = 1 $ \times$ \$10 = \$10. And the utility of the pattern $ (ad) $ is calculated as $ u(ad, T_{1}) $ = $ u(a, T_{1}) $ + $ u(d, T_{1}) $ = $ q(a, T_{1}) \times pr(a)$ + $ q(d, T_{1})\times pr(d) $ = $ 2 \times \$6 $ + $ 2 \times \$7 $ = 26. Correspondingly, the utility of the pattern $ (acd) $ is calculated as $ u(acd) $ = $ u(acd, T_{1}) $ + $ u(acd, T_{7}) $ + $ u(acd, T_{10}) $ = \$36 + \$62 + \$39 = \$137.

\begin{definition}
	\label{def_7}
	\rm The transaction utility of a transaction $ T_{q} $ is denoted as $ tu(T_{q}) $ and defined as: $tu(T_{q})$ = $ \sum_{i_{j}\in T_{q}}u(i_{j}, T_{q})$, where $j$ is the number of items in $T_{q}$. The total utility in a database \textit{D} is the sum of all transaction utilities in \textit{D} and is denoted as $TU$. It is formally defined as: $ TU $ = $\sum _{T_{q}\in D}tu(T_{q})$.
\end{definition}

For example, in Table \ref{db:example}, $ tu(T_{3}) $ = $ u(b, T_{3}) $ + $ u(c, T_{3}) $ + $ u(e, T_{3}) $ = \$2 + \$10 + \$15 = \$27. And the transaction utilities of transactions \textit{T}$ _{1} $ to \textit{T}$ _{10} $ are respectively calculated as \textit{tu}(\textit{T}$ _{1} $) = \$36, \textit{tu}(\textit{T}$ _{2} $) = \$15, \textit{tu}($T_{3} $) = \$27, \textit{tu}($T_{4}$) = \$38, \textit{tu}($T_{5}$) = \$42, \textit{tu}($T_{6}$) = \$9, \textit{tu}($T_{7}$) = \$62, \textit{tu}($T_{8}$) = \$23, \textit{tu}($T_{9} $) = \$34, and \textit{tu}(\textit{T}$ _{10} $) = \$39. The total utility in \textit{D} is the sum of all transaction utilities in \textit{D}, which is calculated as: $TU$ = (\$36 + \$15 + \$27 + \$38 + \$42 + \$9 + \$62 + \$23 + \$34 + \$39) = \$325.

\begin{definition}
	\label{def_9}
	\rm A pattern \textit{X} in a database is a high-utility pattern (HUP) iff its utility is no less than the minimum utility threshold (\textit{minUtil}) multiplied by the total utility of database, that is:
	\begin{equation}
		HUP\leftarrow\{X|u(X)\geq minUtil \times TU\}.
	\end{equation}
\end{definition}

\begin{definition}
	\label{def_10}
	\rm A pattern \textit{X} in a database  $ D $ is said to be a recent high-utility pattern (RHUP) if it satisfies two conditions: 
	\begin{equation}
	    RHUP\leftarrow\{X|r(X) \geq minRe \cap u(X)\geq minUtil \times TU\},
	\end{equation}
 where $ minUtil $ and $ minRe $ are respectively called the minimum utility threshold and  the minimum recency threshold. These two parameters are specified by the user according to his/her preference and prior knowledge.
\end{definition}

Although the definition of HUP is similar to that of RHUP, a key difference is that RHUPs are discovered in databases where transactions are sorted by ascending order of timestamps and that this information is used to evaluate how recent patterns are. Note that a different transaction ordering would result in a different set of RHUPs. The ascending order of timestamps is adopted because based on the definitions of RHUPs, recent patterns are considered to be more interesting than old ones. 

Next, we analyze the relationship between HUPs and RHUPs. For the running example, assume that $ minRe $ and $ minUtil $ are respectively set to 1.50 and 10\%.  The pattern $ (abd) $ is a HUP since its utility is $ u(abd) $ = \$57 $> (minUtil \times TU$ = \$32.5), but not a RHUP since its recency is $ r(abd) $ (= 0.5314 $<$ 1.5). All HUPs are shown in Table \ref{table:RHUPs}, where RHUPs are colored in red.

\begin{table}[ht]
	\setlength{\abovecaptionskip}{0pt}
	\setlength{\belowcaptionskip}{0pt} 	
	\caption{Derived HUPs and RHUPs}
	 \label{table:RHUPs} 
	\centering
    \begin{tabular}{|c|c|c|c|c|c|}
		\hline
		\textbf{Pattern} & $\boldsymbol{{r(X)}}$ & $\boldsymbol{{u(X)}}$  & \textbf{Pattern} &	$\boldsymbol{{r(X)}}$	& $\boldsymbol{{u(X)}}$  \\ \hline  \hline	
		\textcolor{red}{\textbf{(\textit{a})}} &	2.9145 &	\$66  &       $ (ce) $ &	1.2405 &	\$45     \\ \hline
		\textcolor{red}{\textbf{(\textit{c})}} &	3.6235 &	\$100	&   	$ (de) $ &	1.3414 &	\$57   	\\ \hline
		\textcolor{red}{\textbf{(\textit{d})}} &	4.3626 &	\$112 &       $ (abd) $ &	0.5314 &	\$37     \\ \hline
		\textcolor{red}{\textbf{(\textit{e})}} &	2.3624 &	\$35	&   \textcolor{red}{\textbf{(\textit{acd})}} &	1.9048 &	\$137  	\\ \hline		
		\textcolor{red}{\textbf{(\textit{ac})}} &	2.3831 &	\$140 &   	$ (ade) $ &	0.5314 &    \$39    \\ \hline
		\textcolor{red}{\textbf{(\textit{ad})}} &	2.4362 &	\$111 &       $ (bde) $ &	0.5314 &	\$36     \\ \hline
		\textcolor{red}{\textbf{(\textit{bd})}} &	1.6479 &	\$69  &   	$ (cde)$  & 0.81   &	\$34     \\ \hline	
		\textcolor{red}{\textbf{(\textit{be})}} &	1.5524 &	\$34  &   	$ (abde)$ & 0.5314 &	\$42    \\ \hline	
		\textcolor{red}{\textbf{(\textit{cd})}} &	2.7148 &	\$119 &  	              &        &  \\ \hline 	
    \end{tabular}
\end{table}

\subsection{Problem Statement}

\begin{definition}
	\rm  Given a quantitative transactional database (\textit{D}), a \textit{ptable}, a user-specified minimum recency threshold (\textit{minRe}) and a minimum utility threshold (\textit{minUtil}), the goal of RHUPM is to efficiently identify the complete set of RHUPs by considering both the recency and utility constraints. Thus, the problem of RHUPM is to find the complete set of RHUPs, in which the utility and recency of each pattern \textit{X} are respectively no less than $ minUtil \times TU $ and \textit{minRec}. 
\end{definition}

Hence, the goal of RHUPM is to efficiently enumerate all RHUPs in a database, while considering both the recency and utility constraints. To achieve this goal, and reduce the size of the search space, items that cannot generate recent patterns and high-utility patterns should be excluded from the mining process. According to \cite{lin2015efficient}, the concept of UDHUP is defined as below:

\begin{definition}
	\label{def_UDHUP}
	\rm The total utility of the database with its past lifetime to the current one is to sum transaction utilities from transactions $\beta$ to $n$, which is denoted as:  $TU_{[\beta,n]} $ =  $ \sum _{q=\beta}^{n} tu(T_{q}) $ \cite{lin2015efficient}, in which $ 1 \leq \beta \leq n $, and $n$ is the number of transactions in $D$. A pattern \textit{X} in a database  $ D $ is called an up-to-date high-utility pattern (UDHUP) \cite{lin2015efficient} if:
	\begin{equation}
	     UDHUP \leftarrow \{X|u(X)_{[\beta,n]} \geq minUtil \times TU_{[\beta,n]}\},
	\end{equation}
	with its lifetime from $ \beta $ to $n$, which can be thus represented as  \{$X$: $u(X)$, $[\beta, n]$\}.
\end{definition}

%%%  PHILIPPE : I did not see a clear definition of UDHUP... 

Thus, UDHUP mining aim at discovering not only HUPs in the entire databases but also the up-to-date HUPs within its lifetime from the past timestamp to the current one, the so-called high-utility of an UDHUP relies on the summation utility of the transactions which appear in the recent interval w.r.t. lifetime. Based on the above analysis, it can be seen that the proposed RHUP is quite different from the UDHUP, the relationship between HUP, UDHUP, and RHUP is \textit{RHUP} $ \subseteq $ \textit{HUP} $ \subseteq $ \textit{UDHUP}, as illustrated in Fig. \ref{relationships}.

The benefit of RHUPM over HUPM is that the former can extract less but more valuable patterns from temporal databases for decision making by considering the recency of patterns. The recent high-utility patterns can be more useful than the frequent patterns in many real-world applications, such as decision making in intelligent systems (e.g., stock investment) \cite{lin2015discovering}, condition monitoring \cite{batal2012mining}, event detection \cite{batal2012mining}, pattern-based prediction and classification \cite{cheng2007discriminative}.

\begin{figure}[hbtp]	
	\centering
	\includegraphics[trim= 0 0 5 0,clip,scale=0.25]{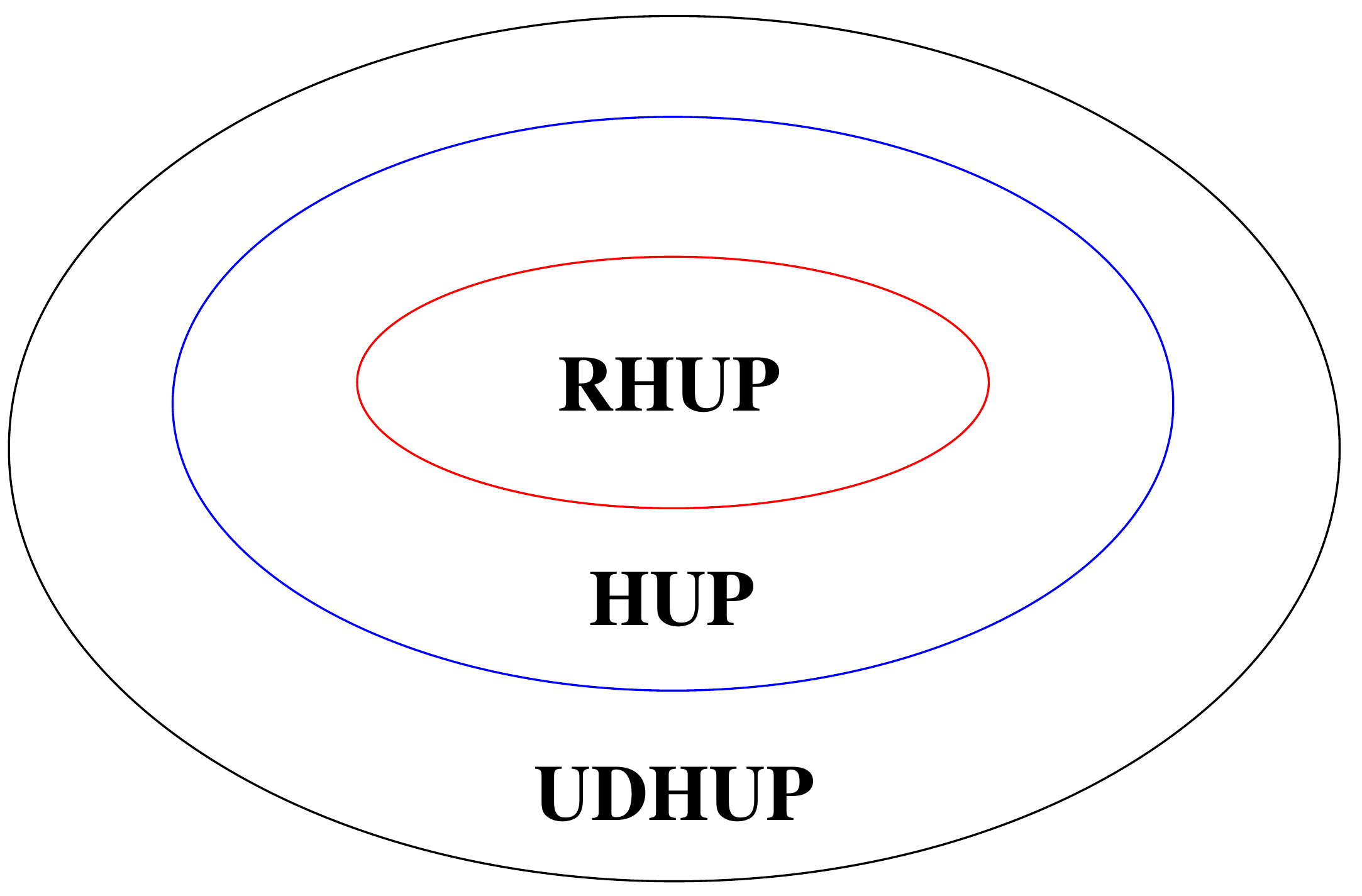}
	\caption{The relationship between HUP, UDHUP, and RHUP.}
	\label{relationships}
\end{figure}

%%%%%%%%%%%%%%%%%%%%%%%%%%%%%%%%%%%%%%%%%%%%%%%%%%%%%%%%
%%%%%%%%%%%%%%%%%%%    RUP ALGORTHM  %%%%%%%%%%%%%%%%%%
%%%%%%%%%%%%%%%%%%%%%%%%%%%%%%%%%%%%%%%%%%%%%%%%%%%%%%%%
\section{Proposed RUP Algorithm for Mining RHUPs}

In this section, we investigate the properties of RHUPs, propose a recency-utility list (RU-list) structure, and present an efficient tree-based utility mining algorithm named RUP. RUP mines RHUPs using a novel recency-utility (RU)-tree, RU-list structure,  and two downward closure properties of RHUPs. Details are given thereafter.

\subsection{The Proposed RU-tree Structure}

The RU-tree is a tree-based data structure designed for mining the set of RHUPs. Its construction can be done by scanning the database once. The RU-tree maintains information about the recency and utility values of items so that RHUPs can be discovered efficiently. To avoid missing RHUPs, the recency and utility values are stored for each occurrence of a pattern, in the RU-tree.

\begin{definition}[Total order $\prec$ on items]
	\rm Assume that there exists a total order $\prec$ on items, defined as the ascending order of transaction-weighted utilization (\textit{TWU}) \cite{liu2005two}. Let $ TWU(X)$ denote the sum of all transaction utilities containing $ X $ and defined as $ TWU(X) $ = $\sum_{X\subseteq T_{q}\wedge T_{q}\in D}tu(T_{q})$ \cite{liu2005two}.
\end{definition}

For example, the ascending order of \textit{TWU} for items in the running example is  \textit{TWU}($e$): \$112 $ < $ \textit{TWU}($b$): \$116 $ < $ \textit{TWU}($a$): \$217 $ < $ \textit{TWU}($c$): \$236 $ < $ \textit{TWU}($d$): \$251. Hence, the total order $\prec$ on items in the RU-tree is $ e \prec b \prec a \prec c \prec d $.

\begin{definition} [Recency-utility tree, RU-tree]
	\rm A recency-utility tree (RU-tree) is a variation of set-enumeration tree \cite{rymon1992search}, it incorporates both recency and utility factors and the total order $\prec$ is applied on items.
\end{definition}

\begin{definition}[Extension nodes in the RU-tree]
	\rm  The extensions of a pattern $X$ ($X$ is presented as a node in the RU-tree)  are obtained by appending one ore more item $y$ to $X$ such that $y$ succeeds all items  in $X$ according to the total order $\prec$. Thus, the extensions of $X$ are its descendant nodes.
\end{definition}

Consider the RU-tree of the running example, the  extension (descendant) nodes  of the node $(ea) $ are $ (eac) $, $ (ead)$,  and $ (eacd) $, while the supersets of the node  (\textit{ea}) are (\textit{eba}), (\textit{eac}), (\textit{ead}), (\textit{ebac}), (\textit{ebad}), (\textit{eacd}), and (\textit{ebacd}). Hence, the set of extension nodes of a node is a subset of its supersets.

Note that the RU-tree is just a conceptual representation of the search space of our addressed problem. Indeed, it does not need to be a physical tree in memory. Each node in a RU-tree is a set in the power set of $I$, where the root node is the empty set. Each node in a RU-tree stores the following information: a \textit{\underline{pattern}}, its \textit{\underline{recency-utility-list}} (\textit{\underline{RU-list}}), and its \textit{\underline{links}}. The \textit{\underline{RU-list}} of a pattern is a vertical data structure, which stores information about the transactions where the pattern appears. The RU-list structure will be presented in the next subsection. Finally, \textit{\underline{links}} contains pointers to the child nodes that have the same pattern as prefix. Nodes representing patterns containing a single item are constructed by scanning the database. Other nodes representing larger patterns are constructed by extracting information from already constructed nodes,  without scanning the database, as it will be explained. The RUP algorithm constructs a RUP-tree  as it explores the search space of patterns in the database. When the algorithm has finished building the RU-tree, all RHUPs in the database have been output.

\begin{lemma}
	\label{lemma1}
	The complete search space of the proposed RUP algorithm for the task of RHUPM can be represented by a RU-tree where items are sorted by ascending order of TWU.
\end{lemma}
\begin{proof}
	According to the definition of a Set-enumeration tree \cite{rymon1992search}, all $2^m - 1$ non-empty patterns that can be formed using the items in $I$ are represented by nodes in the tree, where $m$ is the number of items in $I$. Thus, this structure can be used to systematically enumerate all subsets of $I$. For example, the complete conceptual RU-tree contains all subsets of $ I$ = $\{a, b, c, d, e\} $, and all the patterns that can be obtained by combining these items are represented by nodes in the tree. Thus, all the supersets of the root node (empty set) can be enumerated according to the TWU-ascending order of items by exploring the tree using a depth-first search. This representation is complete and unambiguous; the developed RU-tree can be used to represent the whole search space explored by the proposed algorithm.
\end{proof}

\begin{lemma}
	\label{lemma2}
	The recency of a node in the RU-tree is no less than the recency of any of its child nodes (extensions).
\end{lemma}
\begin{proof}
	Assume that there exists a node $ X^{k-1} $ in the RU-tree, containing $(k-1)$ items. Let $ X^{k} $ be a child node of $ X^{k-1} $. Hence, the pattern $ X^{k} $ contains $k$ items, and has $(k-1)$ items in common with $X^{k-1}$. Since \textit{X}$ ^{\textit{k}-1} $ $\subseteq$ \textit{X$ ^{k} $}, and \textit{TIDs}(\textit{X}$ ^\textit{k}) $ $\subseteq$ \textit{TIDs}(\textit{X} $ ^{{k}-1} )$, it can be proven that: $r(X^{k})$ = $\sum _{X^{k}\subseteq T_{q}\wedge T_{q}\in D}r(X^{k},T_{q})$ 	$\leq \sum _{X^{k-1}\subseteq T_{q} \wedge T_{q}\in D}r(X^{k-1},T_{q})$. Thus,  $r(X^{k}) \leq r(X^{k-1})$,	the recency of a node in the proposed RU-tree is always no less than the recency of any of its extension nodes.
\end{proof}

The designed RU-tree is a compact prefix-tree based representation of the database. It contains all the information required for mining high-utility patterns without scanning the original database. The size of the RU-tree is bounded by, but usually much smaller than, the total number of item occurrences in the database. In general, tree-based mining algorithms such as RUP proceed as follows: they first construct the nodes representing single items in the global tree, then the rest of the tree w.r.t. the search space is explored, then the pruning and mining operations are performed to prune branches of the tree and output the desired set of patterns. In the developed RU-tree structure, only the promising patterns and their RU-lists are needed explored. This structure provides the interesting property that it can be explored/built on-the-fly during the mining process.

Exploring all subsets is time-consuming for large databases. Thus, it is desirable to use pruning strategies to avoid exploring all patterns. In FPM, pruning is carried out using the well-known downward closure property of the support. Since this property  does not hold in HUPM, the traditional \textit{TWDC} property of the \textit{TWU} model was proposed as an alternative to prune the search space \cite{liu2005two,gan2018survey}. Nevertheless, if an algorithm for RHUPM only uses this property to prune the search space, it may be inefficient. The following lemmas present novel ways of pruning unpromising patterns, specific to RHUPM. These lemmas are based on a novel data structure called recency-utility list (RU-list), which compactly stores information about a pattern. Using this structure, it can be ensured that all RHUPs can be derived while pruning numerous unpromising items. The RU-list structure is defined as follows.

\subsection{The RU-list Structure}

The recency-utility list (RU-list) structure is a new vertical data structure, which incorporates the concepts of recency and utility. It is used to store all the necessary information for mining RHUPs. Let there be a pattern \textit{X} and a transaction (or pattern) \textit{T} such that $X\subseteq T$. The set of all items in \textit{T} that are not in \textit{X} is denoted as \textit{T}$\setminus$\textit{X}, and the set of all  items appearing after \textit{X} in \textit{T} (according to the $\prec$ order) is denoted as $ T/X $. For example, consider that \textit{X} = \{\textit{bd}\} and $T$ is the transaction $ T_{5} $ depicted in Table \ref{db:example}.  $ T_{5} $$\setminus$\textit{X} = \{\textit{ae}\}, and $ T_{5}/X $ = \{\textit{e}\}. Thus, $ T/X\subseteq$ \textit{T}$\setminus$\textit{X}. 

\begin{definition}[Recency-Utility list, RU-list]
	\label{def_14}
	\rm The RU-list of a pattern \textit{X} in a database is denoted as \textit{X.RUL}. Inspired by the utility-list \cite{liu2012mining}, RU-list contains an entry (element) for each transaction $ T_{q} $ where \textit{X} appears ($ X \subseteq T_{q} \wedge T_{q} \in D $). An element corresponding to a transaction $T_q$ has four fields: (1) the \textit{\textbf{\underline{tid}}} of \textit{T$ _{q} $}; (2) the recency of \textit{X} in \textit{T$ _{q} $} (\textit{\textbf{\underline{rec}}}); (3) the utility of \textit{X} in \textit{T$ _{q} $} (\textbf{\textit{\underline{iu}}}); and (4) the remaining utility of \textit{X} in \textit{T$ _{q} $} (\textbf{\textit{\underline{ru}}}), in which \textit{ru} is defined as $X.ru(T_{q})=\sum_{i_{j}\in(T_{q}/X)}u(i_{j},T_{q})$.
\end{definition}

For example, the RU-list of item $d$ is shown in Fig. \ref{fig:ru-list} (right). All the important information for discovering RHUPs can be compressed into the designed RU-list structure by performing only one database scan. The proposed RUP algorithm initially builds the RU-lists of single items by scanning the database. Then, the RU-list of any larger $k$-pattern ($k > 1$) is obtained by joining the RU-list of its parent node with the RU-list of an uncle node (which are ($k$-1)-patterns). This process can be applied recursively to obtain the recency and utility information of any pattern without rescanning the database. The detailed procedure for constructing the RU-list of a \textit{k}-pattern is given in Algorithm 3. RU-lists are constructed according to the ascending order of TWU. For the running example, recall that this order is $ e \prec b \prec a \prec c \prec d  $. Hence, the constructed RU-lists for single items are as shown in Fig. \ref{fig:ru-list}.

\begin{figure*}[!hbtp]
	\centering
	\includegraphics[trim = 0 0 0 0,clip,scale=0.45]{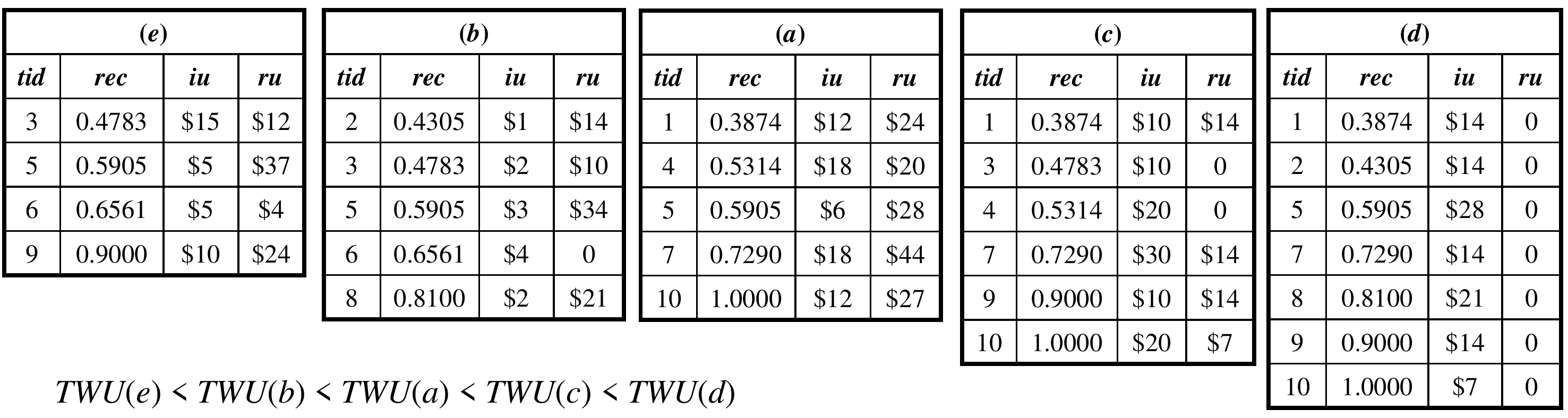}
	\caption{The constructed RU-lists of 1-patterns.}
	\label{fig:ru-list}
\end{figure*}

Using the developed RU-list structure provides several benefits: (1) the information required for mining RHUPs is stored in a lossless structure; (2) the structure is compact and thus does not requires a large amount of memory; and (3) the  important information about  any ($k$-1)-pattern can be quickly obtained by joining the RU-lists of some of its prefix ($k$-1)-patterns. Furthermore, key information can be derived from the RU-list of a pattern, as explained thereafter.

\begin{definition}
	\label{def_16}
	\rm Given a database $D$, the RU-list of a pattern \textit{X} allows deriving the recency of $X$  ($r(X)$), denoted as $ X.RE$, and calculated as:  $X.RE$ = $ \sum_{X\subseteq T_{q} \wedge T_{q}\in D}(X.rec)$.  
\end{definition}

\begin{definition}
	\rm Let the sum of the utilities of a pattern \textit{X} in \textit{D} be denoted as $ X.IU $. Using the RU-list of $X$, $X.IU$ is calculated as follows: $X.IU$ = $ \sum_{X\subseteq T_{q}\wedge T_{q}\in D}(X.iu)$.
\end{definition}

\begin{definition}
	\rm Let the sum of the remaining utilities of a pattern \textit{X} in \textit{D} be denoted as $ X.RU $. Using the RU-list of $X$, $X.RU$ is calculated as: $X.RU$ = $ \sum_{X\subseteq T_{q}\wedge T_{q}\in D}(X.ru)$.
\end{definition}

For example, consider pattern $(b)$ in Table \ref{db:example},  $(b)$ appears in transactions having the \textit{TIDs} $\{2, 3, 5, 6, 8\}$. The value $b.RE$ is calculated as (0.4305 + 0.4783 + 0.5905 + 0.6561 + 0.8100) = 2.9654, $b.IU$ is calculated as (\$1 + \$2 + \$3 + \$4 + \$2) = \$12, and $b.RU$ is calculated as (\$14 + \$10 + \$34 + \$0 + \$21) = \$79. Consider the pattern $(bd)$. It appears in transactions with \textit{TIDs}  $\{2, 5, 8\}$. The value $ (bd).RE$ is calculated as (0.4305 + 0.5905 + 0.8100) = 1.831, $ (bd).IU $ = (\$1 + \$14) + (\$3 + \$28) + (\$2 + \$21) = \$69, and $ (bd).RU $ = \$0.

\subsection{The Global and Conditional Downward Closure Properties}

Based on the RU-list structure and some properties of the recency and utility measures, five novel search space pruning strategies are designed. Those strategies allow the RUP algorithm to prune unpromising patterns early to reduce the search space, and thus improve its efficiency.

\begin{lemma}
	\label{lemma3}
	The actual $utility$ of a node/pattern in the RU-tree is 1) less than, 2) equal to, or 3) greater than the utility of any of its extension (descendant) nodes.
\end{lemma}

Given the above lemma, the downward closure property of ARM cannot be used in HUPM to mine HUPs. In traditional HUPM, the \textit{TWDC} property \cite{liu2005two} was proposed to reduce the search space. However, this property does not consider the recency measure. To address this issue, the following lemmas and theorems are proposed based on the RU-list structure and properties of the recency and utility measures. These lemmas and theorems are designed to be used with the designed RU-tree structure to prune the search space.

\begin{definition}
	\label{def_11}
	\rm A pattern \textit{X} in a database  $ D $ is said to be a recent high transaction-weighted utilization pattern (RHTWUP) if it satisfies two conditions: 1) $ r(X) \geq minRe $; 2) $ TWU(X) \geq minUtil \times TU $, where $ minUtil $ is the minimum utility threshold and $ minRe $ is the minimum recency threshold.
\end{definition}

\begin{theorem}[Global downward closure property, GDC property]
	\label{theorem1}
	\textit{Let  there be a $k$-pattern (node) $X^k $ in the RU-tree, and $ X^{k-1}$ be a  ($k$-1)-pattern (node)  such that its first $k$-1 items are the same as $ X^k $. The global downward closure (GDC) property guarantees that: $ TWU(X^k) \leq TWU(X^{k-1}) $ \cite{liu2005two} and $ r(X^k) \leq r(X^{k-1}) $.}
\end{theorem}
\begin{proof}
	Let \textit{X}$ ^{\textit{k}-1} $ be a (\textit{k}-1)-pattern, and \textit{X$ ^{k} $} be  a superset of  \textit{X}$ ^{\textit{k}-1}$, containing $k$ items.
 Since \textit{X}$ ^{\textit{k}-1} $ $\subseteq$ \textit{X$ ^{k} $},		$TWU(X^{k})$ = $\sum _{X^{k}\subseteq T_{q}\wedge T_{q}\in D}tu(T_{q})$
		$\leq \sum _{X^{k-1}\subseteq T_{q}\wedge T_{q}\in D}tu(T_{q})$. Thus, we have $ TWU(X^{k}) $ $\leq TWU(X^{k-1})$ \cite{liu2005two}.	Besides, it can be found that \textit{r}(\textit{X}$ ^{k-1} $) $\geq$ \textit{r}(\textit{X$ ^{k} $}). Therefore, if \textit{X$ ^{k} $} is a RHTWUP, any  subset \textit{X}$ ^{k-1} $ of \textit{X}$ ^{\textit{k}} $  is also a RHTWUP.
\end{proof}

\begin{theorem}[RHUPs $\subseteq$ RHTWUPs]
	\label{theorem2} 
	A RU-tree is a Set-enumeration tree where the total order $\prec$ is applied. Assuming that order, it follows that RHUP  $\subseteq$ RHTWUP, which means that if a pattern is not a RHTWUP,  none of its supersets is a RHUP.
\end{theorem}

\begin{proof}
	Let $X^{k}$ be an $k$-pattern such that $ X^{k-1} $ is a subset of $ X^{k} $. 
	Based on the property of \textit{TWU} concept \cite{liu2005two}, we know that  $ u(X)\leq TWU(X) $. Moreover, by Theorem \ref{theorem2}, $ r(X^{k}) \leq r(X^{k-1})$ and $ TWU(X^{k}) \leq TWU(X^{k-1})$. Thus, if $ X^{k} $ is not a RHTWUP, none of its supersets is a RHUP. 
\end{proof}

\begin{lemma}
	\label{lemma4}
	The TWU of any node in the RU-tree is greater than or equal to the utility of any  of its descendant nodes, and more generally greater than the utility of  any of its supersets (which may not be its descendant nodes).
\end{lemma}

\begin{proof}
	Let $X^{k-1}$ be a node in the RU-tree, and $X^{k}$ be a children (extension) of $X^{k-1}$. According to the \textit{TWU} concept  \cite{liu2005two}, the relationship $TWU(X^{k-1}) \geq TWU(X^k)$ holds, and thus this lemma holds. 
\end{proof}

\begin{theorem}
	\label{theorem3} 
	In the RU-tree, if the TWU of a tree node $X$ is less than  $ minUtil \times TU $, $X$ is not a RHUP, and all its supersets (not only its descendant nodes) are not RHUPs. 
	%}
\end{theorem}

\begin{proof}
	This theorem directly follows from Theorem \ref{theorem2}. 
\end{proof}	

To ensure that the RUP algorithm discovers all RHUPs, it utilizes the utility and remaining utility of a pattern to calculate an overestimation of the utility of any of its descendant nodes in the RU-tree. This upper bound by defined in the following lemma.

\begin{lemma}
	\label{lemma5}
	Let there be a pattern $X^k$ appearing in a database $D$. Furthermore, let $X^{k-1}$ be a child node of $X^k$ in the RU-tree. The sum of $X^{k-1}.IU + X^{k-1}.RU$ (calculated using the  RU-list of $X^{k-1}$) is an upper bound on the utility of $X^k$, i.e., $ X^{k}.IU \leq X^{k-1}.IU + X^{k-1}.RU$.
\end{lemma}

\begin{proof} 
    According to \cite{liu2012mining,lin2016efficient}, this lemma holds. The utility of $ X^{k} $ in \textit{D} is always less than or equal to the sum of actual utility and remaining utility of $ X^{k-1} $ (w.r.t. $X^{k-1}.IU$ + $X^{k-1}.RU$).
\end{proof}	

\begin{theorem}[Conditional downward closure property, CDC property]
	\label{theorem4} 
	For any node $X$ in the RU-tree, the sum of $X.IU$ and $X.RU$ in its RU-list is larger than or equal to the utility of any of its descendant nodes (extensions). Thus, this upper-bound provides a conditional anti-monotonicity property for pruning unpromising patterns in the RU-tree.
\end{theorem}

\begin{proof}
	According to the definition of a Set-enumeration tree \cite{rymon1992search} and Lemma \ref{lemma5}, this theorem holds. Thus, the utility of a pattern appearing in a database \textit{D} is always less than or equal to the sum of the utilities and remaining utilities of any of its ancestor nodes.
\end{proof}

The above lemmas and theorems can be used to reduce the search space for mining RHUPs, while ensuring that no RHUPs will be missed. We incorporate the results of all these lemmas and theorems into a novel utility-driven mining algorithm, called RUP. The designed GDC and CDC properties guarantee the \textbf{completeness} and \textbf{correctness} of the proposed RUP algorithm. By utilizing the GDC property, the algorithm only needs to initially construct the RU-lists of promising patterns (the \textit{RHTWUPs} containing a single item, denoted as  \textit{RHTWUPs}$^1 $) to then explore the other patterns recursively. Furthermore, the following pruning strategies are proposed in the RUP algorithm to speed up the discovery of RHUPs.

\subsection{The Proposed Pruning Strategies}

As mentioned in subsection 4.1, the size of the  search space for the problem of RHUPM is $2^m$ - 1 patterns (where $ m $ is the number of items in $I$), by systematically enumerating all subsets of $I$. If no powerful pruning strategies are used, the search space for discovering the desired RHUPs is huge. Therefore, inspired by previous studies, several efficient pruning strategies are integrated in the designed RUP algorithm to prune unpromising patterns early. Those strategies are based on the above lemmas and theorems, and generally allows to explore a much smaller part of the search space. Details are given thereafter.

\begin{strategy}[\textit{TWU} pruning strategy]
  \rm	By scanning the database once, the recency and \textit{TWU} values of each item can be obtained. If the \textit{TWU} of an item $i$ (i.e., $TWU(i)$) and the sum of all the recencies of $i$ (i.e., $r(i)$) do not satisfy the two conditions of a RHTWUP, this item can be directly pruned, because none of its supersets is a RHUP.
\end{strategy}

\begin{example}
	By Theorem \ref{theorem2}, all subsets of a RHTWUP are also RHTWUPs. Since \{\textit{r}($e$): 2.3624, \textit{TWU}($e$): \$112; \textit{r}($b$): 2.9654, \textit{TWU}($b$): \$116; \textit{r}($a$): 2.9145, \textit{TWU}($a$): \$217; \textit{r}($c$): 3.6235, \textit{TWU}($c$): \$236; \textit{r}($d$): 4.3626, \textit{TWU}($d$): \$251\}, all five items satisfy the two necessary conditions to be RHTWUPs. The 1-pattern RHTWUPs are shown in Table \ref{table:1_RHTWUPs}. Any of their supersets can be a RHUP and thus the branches in the RU-tree corresponding to each of these patterns cannot be pruned.
\end{example}

\begin{table}[ht]
	\setlength{\abovecaptionskip}{0pt}
	\setlength{\belowcaptionskip}{0pt} 	
	\caption{Derived 1-pattern RHTWUPs}
	\label{table:1_RHTWUPs} 
	\centering
	\begin{tabular}{|c|c|c|c|c|c|}
		\hline
		\textbf{Pattern} & $\boldsymbol{ {r(X)} }$ & $\boldsymbol{{TWU(X)}}$  & $\boldsymbol{{u(X)}}$ &	$\boldsymbol{{r(X)}}$ & $\boldsymbol{{TWU(X)}}$  \\ \hline \hline	
		$ (a) $ &	2.9145 &	\$217  &       $ (d) $ &	4.3626 &	\$251     \\ \hline
		$ (b) $ &	2.9654 &	\$116	&   	 $ (e) $ &	2.3624 &	\$112   	\\ \hline
		$ (c) $ &	3.6235 &	\$236 &        $     $ &	 &	     \\ \hline	
	\end{tabular}
\end{table}

\begin{strategy}[\textit{recency} pruning strategy]
	\rm  When traversing the RU-tree using a depth-first search, if the sum of all the recency values of a node $X$ (i.e., $ X.RE$) in its constructed RU-list is less than the minimum recency threshold, then none of its child nodes is a RHUP. 
\end{strategy}

\begin{example}
	As shown in Fig. \ref{fig:RU_tree2}, the \textit{recency} of the pattern $(ea)$ is  no less than that of its extensions $(eac)$, $(ead)$ and $(eacd)$, which are respectively calculated as $r(ea)$ = 0.5905, $r(eac)$ = 0, $ r(ead) $ = 0.5905, and $ r(eacd) $ = 0. Thus, the \textit{recency}  of $(eac)$, $(ead)$ and $(eacd)$ are all less than or equal to the \textit{recency} of $(ea)$. The pattern $(ea)$ is neither a RHTWUP nor a RHUP since $ TWU(ea)$ = \$42$ >$ \$32.5, but $ r(ea)$ = 0.5905 $<$ 1.5. Thus, when the designed algorithm processes the supersets of $(ea)$, it can skip processing all patterns that are descendant (extension) nodes of $(ea)$. This means that the patterns $ \{(eac), (ead), (eacd)\} $ can be directly pruned (the sub-branches of $X$ do not need to be explored).  
\end{example}

\begin{figure}[hbtp]	
	\centering
	\includegraphics[trim= 0 0 5 0,clip,scale=0.25]{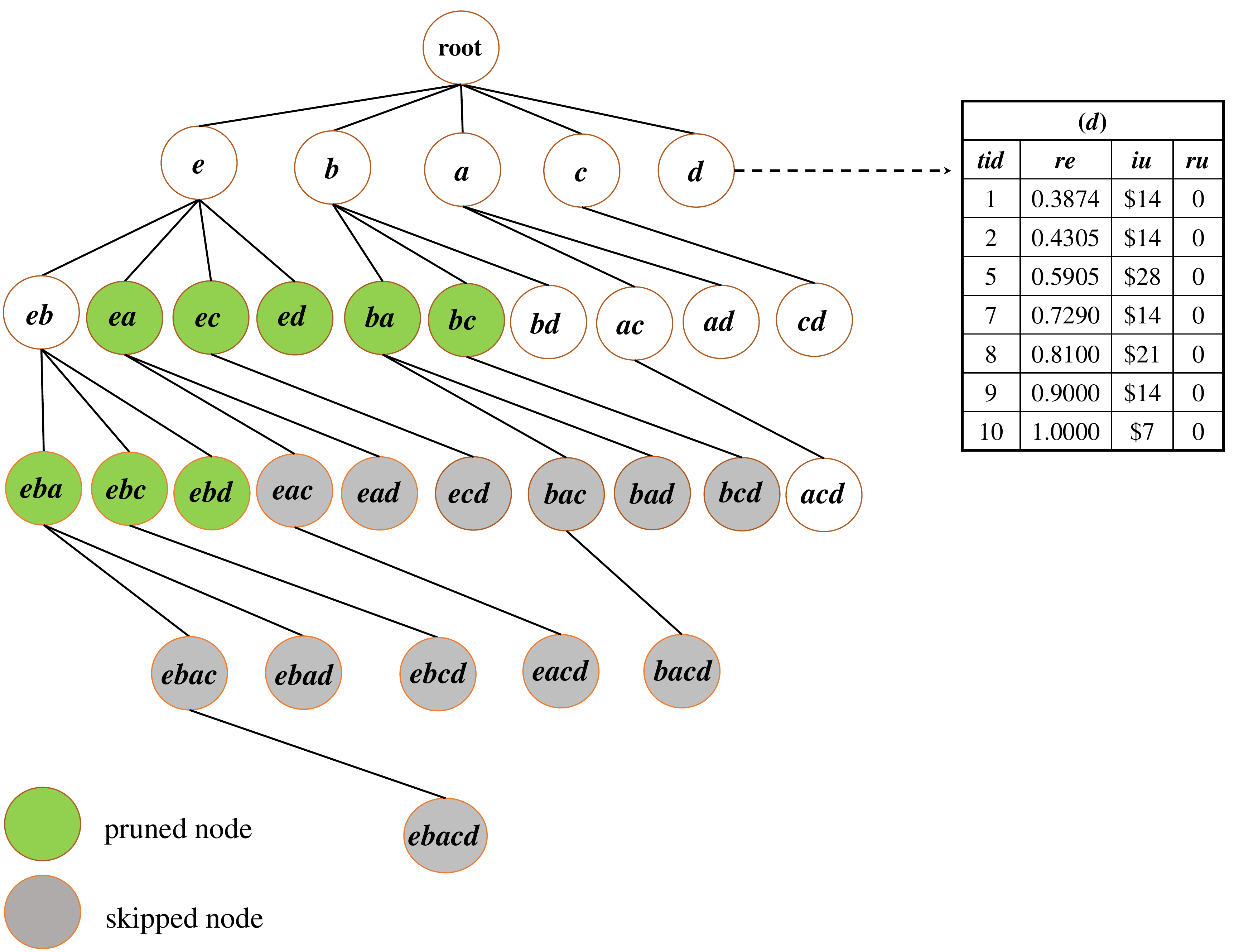}
	\caption{Application of the pruning strategies on the conceptual RU-tree.}
	\label{fig:RU_tree2}
\end{figure}

By Lemma \ref{theorem1} and Theorem \ref{lemma5},  upper bounds on a pattern's utility can be used to reduce the search space, as described in Strategy 3. Furthermore, an effective EUCP (Estimated Utility Co-occurrence Pruning) strategy \cite{fournier2014fhm} is also integrated in the RUP algorithm to speed up the discovery of RHUPs. It is defined as follows: if the \textit{TWU} of a 2-pattern is less than the $minUtil$ threshold, any superset of this pattern is neither a RHTWUP nor a HUP \cite{fournier2014fhm}. According to the definitions of RHTWUP and RUP, this can be applied in the proposed RUP algorithm to further filter unpromising patterns. To effectively apply the EUCP strategy, a structure named Estimated Utility Co-occurrence Structure (\textit{EUCS}) \cite{fournier2014fhm} is built by the proposed algorithm. It is a matrix that stores the \textit{TWU} values of  2-patterns. This structure is employed by strategy \ref{strategy:EUCS}.

\begin{table}[!ht]
	\setlength{\abovecaptionskip}{0pt}
	\setlength{\belowcaptionskip}{0pt} 	
	\centering
	\label{table:EUCS}
	\caption{The constructed $EUCS$}
	\begin{tabular}{|c|c|c|c|c|c|c|}
		\hline
		\textbf{Item} & \textbf{\textit{a}} & \textbf{\textit{b}} & \textbf{\textit{c}} & \textbf{\textit{d}} & \textbf{\textit{e}}  \\ \hline \hline
		\textbf{\textit{b}}    & \$42        & -          & -          & -          & -               \\ \hline
		\textbf{\textit{c}}    & \$175       & \$27        & -          & -          & -            \\ \hline
		\textbf{\textit{d}}    & \$179       & \$80        & \$171        & -          & -            \\ \hline
		\textbf{\textit{e}}    & \$42        & \$78        & \$61        & \$76        & -             \\ \hline
	\end{tabular}
\end{table}

\begin{strategy}[\textit{upper bound on utility} pruning strategy]
	\rm When traversing the RU-tree based on a depth-first search strategy, if the sum of $ X.IU $ and $ X.RU $ of any node $X$ is less than the minimum utility count, any of its child node is not a RHUP, and can thus be directly pruned (the sub-branches of $X$ do not need to be explored).
\end{strategy}

\begin{strategy}[\textit{EUCP} strategy]	
	\rm Let $X$ be a pattern (node) encountered during a depth-first search on a RU-tree. If the \textit{TWU} of a 2-pattern $Y$ $\subseteq$ $X$ according to the constructed \textit{EUCS} is less than the minimum utility threshold, $X$ is neither a RHTWUP nor a RHUP. As a consequence, none of its child nodes is a RHUP. The construction of the RU-lists of $X$ and its children thus does not need to be performed. 
	\label{strategy:EUCS}
\end{strategy}

\begin{example}
	Consider the RU-tree built for the running example, depicted in Fig. \ref{fig:RU_tree2}. The pattern $(bc)$ is not a RHTWUP since $ TWU(bc)$ = \$27 $<$ \$32.5. Thus, by applying the EUCP strategy, the supersets of $(bc)$, including $(bcd)$, $(ebc)$, $(bac)$, $(ebac)$, $(ebcd)$, $(bacd)$, and $(ebacd)$, can be skipped.
\end{example}
	
\begin{strategy}[\textit{RU-list} checking strategy]
  \rm	Let $X$ be a pattern (node) visited during a depth-first search on a RU-tree. After constructing the RU-list of $X$, if $ X.{RUL} $ is empty or the $ X.RE $ value is less than the minimum recency threshold, $X$ is not a RHUP, and none of its child nodes is a RHUP. As a consequence, the  RU-lists of the child nodes of $X$ do not need to be constructed.
\end{strategy}

\begin{example}
Consider the RU-tree depicted in Fig. \ref{fig:RU_tree2} and the pattern $ (eac)$. The RU-list of this pattern is empty since it does not appear in any transactions of the running example database. Thus, RU-lists of descendant (extension) nodes of $ (eac) $ do not need to be constructed and these patterns are skipped. As shown in Fig. \ref{fig:RU_tree2}, the search space of the running example contains ($2^m$ - 1) = $2^5$ - 1 = 31 patterns if all patterns are explored, while the number of nodes visited by the proposed algorithm is 19 (12 nodes are skipped). Among them, 8 nodes are visited and pruned without exploring their descendant nodes since they are not RHUPs according to the pruning strategies. The final set of RHUPs is shown in Table \ref{table:RHUPs}.
\end{example}

\subsection{The RUP Algorithm}

Based on the above properties and pruning strategies, the pseudo-code of the proposed RUP algorithm is presented in Algorithm 1. The RUP algorithm first sets $X.RUL$, $D.RUL$ and $EUCS$ to the empty set (Line 1). Then, RUP scans the database to calculate the $ TWU(i) $ and $ r(i) $ values of each item $ i\in I $ (Line 2) to then identify those that may be RHUPs (Line 3). After sorting $I^*$ according to the total order $\prec$ (the ascending order of \textit{TWU}, Line 4), the algorithm scans $D$ again to construct the RU-list of each 1-item $ i\in I^* $ and build the \textit{EUCS} (Line 5). The RU-lists of all extensions of  each item $ i\in I^* $ are then recursively processed using a depth-first search procedure named RHUP-Search (Line 6) and then the final set of RHUPs is returned (Line 7). It is important to notice that the RU-tree is a conceptual presentation of the search space of the RUP algorithm. RUP only need to construct the original RU-lists of 1-patterns and then generate the $k$-patterns with their RU-lists for determining RHUPs. Thus, the RUP algorithm does not contain details of RU-tree generation. 

%%%%%%%%%%% algorithm 1 %%%%%%%%%%%%%%%

\begin{algorithm}
    \LinesNumbered
	\caption{\rm The RUP algorithm}
	\KwIn{\textit{D}, \textit{ptable}, $ \delta $, $minRe$, $minUtil$.}
	\KwOut{The set of recent high-utility patterns (RHUPs).}
	$ X.RUL \gets \emptyset, D.RUL \gets \emptyset, EUCS \gets \emptyset $\;
	scan $D$ once to calculate the  $ TWU(i) $ and $ re(i) $ values of each item $ i\in I$\;
	find $ I^* $ $ (TWU(i) \geq minUtil \times TU) \wedge (r(i) \geq minRe) $, i.e. $ RHTWUP^1 $\;
	sort $ I^* $ according to the total order $ \prec $\ (ascending order of $TWU$)\;
	scan $D$ to construct the \textit{X.RUL} of each item $ i\in I^{*} $ and build the \textit{EUCS}\;
	\textbf{call RHUP-Search}($\phi, I^*, minRe, minUtil, EUCS$)\;
	\Return \textit{RHUPs}
	\label{algorithm:RUP}
\end{algorithm}
%%%%%%%%%%%%%%%%%%%%%%%%%%%%%%%%%%%%%

The RHUP-Search procedure is given in Algorithm 2. It takes as input a pattern $X$, a set of child nodes of $X$ (\textit{extendOfX}), the $minRe$ and $minUtil$ thresholds and the \textit{EUCS}. This procedure  processes each pattern $X_a$ in the set \textit{extendOfX} to discover RHUPs (Lines 2-5). Two constraints are then applied to determine if child nodes of $X_a$  should be explored by the depth-first search (Lines 6-18). If a pattern is promising, the $ Construct(X, X_a, X_b) $ procedure (cf. Algorithm 3) is called to construct the RU-lists of all 1-extensions of $X_a$ (i.e., \textit{extendOfX}$_a $) (Lines 8-16). Each constructed 1-extension $X_{ab}$ of the pattern $X_a$ is put in the set \textit{extendOfX}$_a $ to be used by the later depth-first search. The \textit{RHUP-Search} procedure is then recursively called to continue the depth-first search (Line 17).

%%%%%%%%%%% algorithm 2 %%%%%%%%%%%%%%%
\begin{algorithm}
	\LinesNumbered
    \label{algorithm:search}
	\caption{The RHUP-Search procedure}
	\KwIn{$X$, ${extendOfX}$, $minRe$, $minUtil$, $EUCS$.}
	\KwOut{The set of RHUPs.}
	\For {each pattern $ X_{a}\in $ \textit{extendOfX}}
	{
		obtain the $ X_{a}.RE $, $ X_{a}.IU $ and $ X_{a}.RU $ values from the built $ X_{a}.RUL $\;
		\If{$ (X_{a}.IU \geq minUtil \times TU) \wedge (X_{a}.RE \geq minRe) $}
		{
			$ RHUPs\leftarrow RHUPs\cup X_{a} $\;	
		}
		\If {$ (X_{a}.IU + X_{a}.RU \geq minUtil \times TU) \wedge (X_{a}.RE \geq minRe) $}
		{
			$ extendOfX_{a}\leftarrow  \emptyset $\;
			\For {each $ X_{b}\in extendOfX $ such that $a \prec b$}
			{
				\If{$ \exists  TWU(ab)\in  EUCS  \wedge  TWU(ab)\geq minUtil \times TU $}
				{
					$ X_{ab}\leftarrow X_{a} \cup X_{b} $\;
					$ X_{ab}.RUL\leftarrow construct(X, X_{a}, X_{b}) $\;
					\If{$ X_{ab}.RUL \not= \emptyset \wedge (X_{a}.RE \geq minRe) $}
					{
						$ extendOfX_{a}\leftarrow extendOfX_{a}\cup X_{ab}.RUL $\;	
					} 			
				}	
			}   		
			\textbf{call RHUP-Search} $(X_{a}, extendOfX_{a}, minRe, minUtil, EUCS)$\;	
		} 
	}
	\Return \textit{RHUPs}
\end{algorithm}

%%%%%%%%%%% algorithm 3 %%%%%%%%%%%%%%%
\begin{algorithm}
	\LinesNumbered
    \label{algorithm:RU-list}
	\caption{The RU-list construction procedure}
	\KwIn{$ X $, a pattern; $ X_{a} $, an extension of $X$ with an item $a$; $ X_{b} $,  an extension of $X$ with an item $b$ ($ a\neq b$).}
	\KwOut{$ X_{ab}.RUL $, the RU-list of a pattern $ X_{ab} $.}
	set $ X_{ab}.RUL\leftarrow \emptyset$\;
	\For {each element $ E_{a}\in X_{a}.RUL $}
	{
		\If {$ \exists E_{a}\in X_{b}.RUL\wedge E_{a}.tid$ == $E_{b}.tid $}
		{
			\If{$ X.RUL\neq \emptyset $}
			{
				find $ E\in X.RUL$, $E.tid$ = $E_{a}.tid $\;			
				$E_{ab} \leftarrow <E_{a}.tid, E_{a}.re, E_{a}.iu+E_{b}.iu-E.iu, E_{b}.ru>$\;
			}
			\Else
			{
				$E_{ab} \leftarrow <E_{a}.tid, E_{a}.re, E_{a}.iu+E_{b}.iu, E_{b}.ru>$\;
			}
			$ X_{ab}.RUL\leftarrow X_{ab}.RUL \cup E_{ab} $\;
		}	
	}
	\Return $ X_{ab}.RUL $
\end{algorithm}
%%%%%%%%%%%%%%%%%%%%%%%%%%%%%%%%%%%%%

In summary, the RUP algorithm utilizes several pruning strategies and calculate tight upper-bounds on the utility of descendant nodes of each visited node, using the RU-list, to prune unpromising patterns. As a result, the search space and the time required to traverse the RU-tree can be  reduced. Based on the designed RU-tree and the top-down depth-first spanning mechanism, the developed RUP algorithm can directly mine all RHUPs in a database by scanning the database twice, and without generating and maintaining candidates. Although RUP adopts the Estimated Utility Co-occurrence Structure (\textit{EUCS}) which is originally proposed in FHM, note that RUP is different from FHM, as shown in the following aspects: 1) the mining task is different; 2) the used data structures (i.e., RU-tree, RU-list) are novel and different from utility-list;  3) the pruning strategies with both recency and utility as constraints are also different from FHM that only utilizes utility factor; and 4) the discovered results (RHUPs vs HUPs) are different.

\textbf{Complexity analysis}. The complexity of RUP is analyzed as follows. To keep consistent, we assume that there are $n$ transactions and $m$ distinct items
in the temper database $D$. RUP requires a single database scan and takes $O(n \times m) $ time in the worst case to calculate the  $ TWU(i) $ and $ re(i) $ values of each item $ i\in I$. Sorting $I^*$ needs $O(mlogm) $ time and constructing $|I^*|$ RU-lists of 1-pattern in $I^*$ needs $O(n \times m)$ time and $O(n\times I^*)$ space. Due to  the length of the longest transaction in the database is \textit{maxL} = $max\{|T_{q}|, T_q \in D\}$ and may be up to $m$. It indicates this longest transaction has all $m$ items such as \textit{maxL} = $m$, in the worst case. Thus,  there is up to $2^m$ - 1 patterns in the search space of RUP. Without loss of generality, we assume that all the $m$ items are promising and their RU-lists have $n$ entries. Thus, the worst case time complexity of the RUP algorithm  is  $O(2^{m} - 1)$ in theory.

\section{Experimental Results}

This section describes substantial experiments conducted to evaluate the effectiveness and efficiency of the proposed RUP algorithm. 
Note that only one study was previously published on the topic of mining recent HUPs. In that study, the UDHUP-apriori and UDHUP-list algorithms were presented \cite{lin2015efficient}. However, as previously mentioned, the concept of UDHUP patterns is quite different from the RHUP patterns considered in this paper. Thus, this experimental evaluation considers the well-known conventional HUPM algorithm (i.e., FHM \cite{fournier2014fhm}) as  benchmark to discover HUPs for evaluating the effectiveness of the proposed RUP algorithm. It is important to notice that RHUPM and traditional HUPM are two different mining tasks. Although the state-of-the-art HUPM algorithms (i.e., FHM \cite{fournier2014fhm}, HUP-Miner \cite{krishnamoorthy2015pruning}, EFIM \cite{zida2017efim}, HMiner \cite{krishnamoorthy2017hminer} and others) are available at the SPMF website\footnote{\url{http://www.philippe-fournier-viger.com/spmf/}}, but they are not suitable to be compared for evaluating the efficiency, i.e., execution time. Due to the requirement of reviewer, RUP was compared by some of these conventional HUPM algorithms. Furthermore, to evaluate the efficiency of the proposed pruning strategies, versions of the RUP algorithm without pruning strategies are also considered in this experimental evaluation. The notation RUP$\rm _{baseline} $,  RUP1, and RUP2 respectively indicates RUP without the pruning strategies 4 and 5, without strategy 5, and with all strategies.

% \footnote{\url{http://fimi.ua.ac.be/data}}

% \footnote{\url{http://fimi.cs.helsinki.fi}}

The experiments were conducted on a personal ThinkPad T470p computer with an Intel(R) Core(TM) i7-7700HQ CPU @ 2.80 GHz 2.81 GHz, 32 GB of RAM, and with the 64-bit Microsoft Windows 10 operating system. Three real-life datasets obtained from the public FIMI repository\footnote{\url{http://fimi.uantwerpen.be/data/}} (retail, chess, and mushroom), foodmart \cite{foodmart}, and two synthetic datasets (T10I4D100K and T40I10D100K) were used in the experiments. 

\begin{itemize}
	\item \textbf{foodmart}: it contains 1,559 distinct items appearing in 21,556 customer transactions, where the average and maximal transaction length are respectively 4 and 11 items. 

	\item \textbf{retail}: this dataset is a large sparse dataset having 88,162 transactions and 16,470 distinct items, representing approximately five months of customer transactions. The average transaction length in this dataset is 10.3 items, and most customers have bought  7 to 11 items per shopping visit. 

	\item \textbf{chess}: this dense dataset contains the legal moves of a chess game, with 75 distinct items, 3,196 transactions, and an average transaction length of 37 items.

	\item \textbf{mushroom}: this dataset contains 119 distinct items, and 8,124 transactions, each having 23 items. This is a dense dataset.

	\item \textbf{T10I4D100K} and \textbf{T40I10D100K}: they were generated using the IBM Quest Synthetic Data Generator \cite{agrawal1994dataset} using various parameter values.

\end{itemize}

A simulation model \cite{liu2005two,tseng2013efficient,gan2017extracting} was developed to generate purchase quantities and unit profit values of items in transactions for all datasets except foodmart which already has real quantities and unit profit values. A log-normal distribution was used to randomly assign quantities in the [1,5] interval, and item profit values in the [1,1000] interval, as in previous work \cite{tseng2013efficient,gan2017extracting}. These datasets have varied characteristics and represents the main types of data typically encountered in real-life scenarios (dense, sparse and long transactions). % Note that the source code and datasets can be downloaded from the SPMF website\footnote{\url{http://www.philippe-fournier-viger.com/spmf/}}. 

\subsection{Pattern Analysis}

The first experiment consisted of comparing the derived HUP and RHUP patterns  using a fixed time-decay threshold $\delta$ to analyze the relationship between these two types of patterns, and  determine whether the proposed RHUPM framework is acceptable. Note that HUPs are discovered using the FHM algorithm, and  RHUPs are found by the proposed RUP algorithm. A recency ratio of high-utility patterns  is defined as: \textit{recencyRatio} = $\frac{|RHUPs|}{|HUPs|} \times 100\% $. Correspondingly, an outdated ratio is defined as \textit{outdatedRatio}  = $\frac{|HUPs - RHUPs|}{|HUPs|} \times 100\% $ = ( 100\% - \textit{recencyRatio}). In general, the outdated patterns do not make sense for decision making, pattern-based prediction and classification. Thus, \textit{recencyRatio} but not \textit{outdatedRatio} can facilitate the process of discovering and validating the patterns by user.

In this experiment, to make the comparison fair,  the $minRe$ threshold was first fixed and the \textit{minUtil} threshold was varied. The parameters were set as follows: foodmart ($ \delta $ = 0.001\%, $ minRe $ = 1), retail ($ \delta $ = 0.01\%, $ minRe $ = 10), chess ($ \delta $ = 0.5\%, $ minRe $ = 60), mushroom ($ \delta $ = 0.1\%, $ minRe $ = 15), T10I4D100K ($ \delta $ = 0.01\%, $ minRe $ = 5), and T40I10D100K ($ \delta $ = 0.05\%, $ minRe $ = 20).  Then, the \textit{minUtil} threshold was fixed and the $minRe$ threshold was varied. The parameters were  set as follows: foodmart ($ \delta $ = 0.001\%, \textit{minUtil} = 0.001\%), retail ($ \delta $ = 0.01\%, \textit{minUtil} = 0.015\%), chess ($ \delta $ = 0.5\%, \textit{minUtil} = 18\%), mushroom ($ \delta $ = 0.1\%, \textit{minUtil} = 7\%), T10I4D100K ($ \delta $ = 0.01\%, \textit{minUtil} = 0.005\%), and T40I10D100K ($ \delta $ = 0.05\%, \textit{minUtil} = 0.2\%). The results in terms of number of HUPs and RHUPs for various \textit{minUtil} values and a fixed  $minRe$ threshold, and for various $minRe$  values and a fixed \textit{minUtil} threshold, are respectively shown in Table \ref{table:patterns1} and Table \ref{table:patterns2}, respectively.

%%%%%%%%%%%%%%%%%%%%%%%%%%%%%%%%%%%%%%%%%%%%%%%%%%Friedman
\begin{table}[htb]
	\fontsize{6.5pt}{10pt}\selectfont
	\centering
	\caption{Number of patterns when \textit{minUtil} is varied and $minRe$ is fixed.}
	\label{table:patterns1}
	\begin{tabular}{ccllllllllll}
		\hline\hline
		\multirow{2}*{\textbf{Dataset}}&
		\multirow{2}*{\textbf{Notation}}
		&\multicolumn{6}{c}{\textbf{\# pattern under various thresholds}}\\
		\cline{3-8}
		& & \textit{test}$_1$ &  \textit{test}$_2$  &  \textit{test}$_3$   &  \textit{test}$_4$  &  \textit{test}$_5$   &  \textit{test}$_6$  \\ \hline

% 492,041	386,509	267,164	165,304	93,467	49,848
% 41,539	38,552	33,188	26,788	20,515	15,238
% 8.44\%	9.97\%	12.42\%	16.21\%	21.95\%	30.57\%
		 & \textit{minUtil} & 0.002\% & 0.003\% & 0.004\%  &  0.005\%  & 0.006\%  & 0.007\%  \\
	\textbf{(a) foodmart}		& HUPs &  492,041 & 	386,509 & 	267,164 & 	165,304 & 	93,467 & 	49,848  \\
	($\delta$: 0.001\%, \textit{minRe}: 1)	& RHUPs &  41,539 & 	38,552 & 	33,188 & 	26,788 & 	20,515 & 	15,238  \\
		& \textit{recencyRatio} &  8.44\% & 	9.97\% & 	12.42\% & 	16.21\% & 	21.95\% & 	30.57\%  \\	    
		\hline

%  2489,385	22,657	15,713	14,126	12,813	11,705
%  7,203	6,965	6,732	6,499	6,299	6,080
%  0.29\%	30.74\%	42.84\%	46.01\%	49.16\%	51.94\%

 & \textit{minUtil} & 0.013\% & 0.014\% & 0.015\%  &  0.016\%  & 0.017\%  & 0.018\%  \\
\textbf{(b) retail} & HUPs &  2489,385	&	22,657	&	15,713	&	14,126	&	12,813	&	11,705 \\
($\delta$: 0.01\%, \textit{minRe}: 10)	& RHUPs & 7,203	6,965	&	6,732	&	6,499	&	6,299	&	6,080   \\
& \textit{recencyRatio} &  0.29\%	&	30.74\%	&	42.84\%	&	46.01\%	&	49.16\%	&	51.94\%  \\	    
\hline

%  198,921	89,933	39,281	16,848	7,141	2,969
%  94,666	48,809	23,798	14,177	5,070	2,249
%  47.59\%	54.27\%	60.58\%	84.15\%	71.00\%	75.75\%

 & \textit{minUtil} & 17\% & 18\% & 19\%  &  20\%  & 21\%  & 22\%  \\
\textbf{(c) chess}	& HUPs & 198,921  &	89,933  &	39,281  &	16,848  &	7,141  &	2,969   \\
 ($\delta$: 0.5\%, \textit{minRe}: 60) & RHUPs &  94,666  &	48,809  &	23,798  &	14,177  &	5,070  &	2,249  \\
& \textit{recencyRatio} & 47.59\%  &	54.27\%  &	60.58\%  &	84.15\%  &	71.00\%  &	75.75\% \\	    
\hline

%  176,549	34,431	34,331	22,121	13,953	7,601
%  79,577	30,868	15,321	9,714	5,745	2,886
%  45.07\%	89.65\%	44.63\%	43.91\%	41.17\%	37.97\%

 & \textit{minUtil} & 5\% & 6\% & 7\%  &  8\%  & 9\%  & 10\%  \\
\textbf{(d) mushroom} 	& HUPs & 176,549 &	34,431 &	34,331 &	22,121 &	13,953 &	7,601  \\
($\delta$: 0.1\%, \textit{minRe}: 15) & RHUPs & 79,577 &	30,868 &	15,321 &	9,714 &	5,745 &	2,886  \\
& \textit{recencyRatio} &  45.07\% &	89.65\% &	44.63\% &	43.91\% &	41.17\% &	37.97\%  \\	    
\hline

%  4,168,007	310,769	148,090	114,938	94,448	80,939
%  52,230	50,841	49,495	48,158	46,921	45,603
%  1.25\%	16.36\%	33.42\%	41.90\%	49.68\%	56.34\%

 & \textit{minUtil} & 0.005\% & 0.006\% & 0.007\%  &  0.008\%  & 0.009\%  & 0.010\%  \\
\textbf{(e) T10I4D100K}	& HUPs &  4,168,007 &	310,769 &	148,090 &	114,938 &	94,448 &	80,939  \\
 ($\delta$: 0.01\%, \textit{minRe}: 5) & RHUPs &  52,230 &	50,841 &	49,495 &	48,158 &	46,921 &	45,603   \\
& \textit{recencyRatio} & 1.25\% &	16.36\% &	33.42\% &	41.90\% &	49.68\% &	56.34\%  \\	    
\hline

% 790,225	143,591	35,676	6,275	2,654	1,429
% 80,191	17,936	7,579	4,541	2,643	1,427
% 10.15\%	12.49\%	21.24\%	72.37\%	99.59\%	99.86\%

 & \textit{minUtil} & 0.15\% & 0.20\% & 0.25\%  &  0.30\%  & 0.35\%  & 0.40\%  \\
\textbf{(f) T40I10D100K} 	& HUPs &  790,225 &	143,591 &	35,676 &	6,275 &	2,654 &	1,429  \\
($\delta$: 0.05\%, \textit{minRe}: 20) & RHUPs & 80,191 &	17,936 &	7,579 &	4,541 &	2,643 &	1,427  \\
& \textit{recencyRatio} & 10.15\% &	12.49\% &	21.24\% &	72.37\% &	99.59\% &	99.86\% \\	    
\hline

		\hline	\hline
	\end{tabular}
\end{table}
%%%%%%%%%%%%%%%%%%%%%%%%%%%%%%%%%%%%%%%%%%%%%%%%%%%%%%

%%%%%%%%%%%%%%%%%%%%%%%%%%%%%%%%%%%%%%%%%%%%%%%%%%
\begin{table}[htb]
	\fontsize{6.5pt}{10pt}\selectfont
	\centering
	\caption{Number of patterns when $minRe$ is varied and \textit{minUtil} is fixed.}
	\label{table:patterns2}
	\begin{tabular}{ccllllllllll}
		\hline\hline
		\multirow{2}*{\textbf{Dataset}}&
		\multirow{2}*{\textbf{Notation}}
		&\multicolumn{6}{c}{\textbf{\# pattern under various thresholds}}\\
		\cline{3-8}%\morecmidrules
		& & \textit{test}$_1$ &  \textit{test}$_2$  &  \textit{test}$_3$   &  \textit{test}$_4$  &  \textit{test}$_5$   &  \textit{test}$_6$  \\ \hline

%  492,041	492,041	492,041	492,041	492,041	492,041
%  492,041	386,976	267,034	154,384	41,539	41,539
%  100.00\%	78.65\%	54.27\%	31.38\%	8.44\%	8.44\%
		
		 & \textit{minRe} & 0.80 & 0.85 & 0.90  &  0.95  & 1.00  &  1.05  \\
		\textbf{(a) foodmart}	& HUPs &  492,041 &	492,041	 & 492,041 &	492,041 &	492,041 &	492,041   \\
	($\delta$: 0.001\%, \textit{minUtil}: 0.001\%)	& RHUPs & 492,041 &	386,976 &	267,034	 & 154,384 &	41,539 &	41,539  \\
		& \textit{recencyRatio} &  100.00\% &	78.65\% &	54.27\% &	31.38\% &	8.44\% &	8.44\%  \\	    
		\hline

%  15,713	15,713	15,713	15,713	15,713	15,713
%  11,444	10,307	9,233	8,320	7,456	6,732
%  72.83\%	65.60\%	58.76\%	52.95\%	47.45\%	42.84\%
		
		 & \textit{minRe} & 5 & 6 & 7  & 8  & 9  & 10 \\
		\textbf{(b) retail}	& HUPs &  15,713 &	15,713 &	15,713 &	15,713 &	15,713 &	15,713 \\
		($\delta$: 0.01\%, \textit{minUtil}: 0.015\%) & RHUPs & 11,444 &	10,307 &	9,233 &	8,320 &	7,456 &	6,732 \\
		& \textit{recencyRatio} & 72.83\% &	65.60\% &	58.76\% &	52.95\% &	47.45\% &	42.84\% \\	    
		\hline

%  89,933	89,933	89,933	89,933	89,933	89,933
%  89,383	79,253	48,809	18,873	3,949	446
%  99.39\%	88.12\%	54.27\%	20.99\%	4.39\%	0.50\%
		
	 & \textit{minRe} & 20 & 40 & 60  &  80  & 100  &  120  \\
		\textbf{(c) chess}	& HUPs & 89,933 &	89,933 &	89,933 &	89,933 &	89,933 &	89,933 \\
		($\delta$: 0.5\%, \textit{minUtil}: 18\%)	& RHUPs & 89,383 &	79,253 &	48,809 &	18,873 &	3,949 &	446  \\
		& \textit{recencyRatio} & 99.39\% &	88.12\% &	54.27\% &	20.99\%	 & 4.39\%  &	0.50\% \\	    
		\hline

%  34,331	34,331	34,331	34,331	34,331	34,331
%  33,762	31,862	15,321	14,618	13,766	12,892
%  98.34\%	92.81\%	44.63\%	42.58\%	40.10\%	37.55\%
		
	 & \textit{minRe} & 5  &  10  & 15  &  20   &  25  &  30  \\
		\textbf{(d) mushroom}	& HUPs & 34,331 &	34,331 &	34,331 &	34,331 &	34,331 &	34,331   \\
		($\delta$: 0.1\%, \textit{minUtil}: 7\%)	& RHUPs & 33,762 &	31,862 &	15,321 &	14,618 &	13,766 &	12,892   \\
		& \textit{recencyRatio} & 98.34\% &	92.81\% &	44.63\% &	42.58\% &	40.10\% &	37.55\%  \\	    
		\hline

%  4,168,007	4,168,007	4,168,007	4,168,007	4,168,007	4,168,007
%  4,050,086	112,383	79,782	61,910	52,230	45,276
%  97.17\%	2.70\%	1.91\%	1.49\%	1.25\%	1.09\%
		
	& \textit{minRe} & 1 & 2 & 3  &  4  & 5  & 6  \\
		\textbf{(e) T10I4D100K} 	& HUPs &  4,168,007 &	4,168,007 &	4,168,007 &	4,168,007 &	4,168,007 &	4,168,007   \\
		($\delta$: 0.01\%, \textit{minUtil}: 0.005\%)	& RHUPs &  4,050,086 &	112,383 &	79,782 &	61,910 &	52,230 &	45,276  \\
		& \textit{recencyRatio} & 97.17\% &	2.70\% &	1.91\% &	1.49\% &	1.25\% &	1.09\% \\	    
		\hline

%  143,591	143,591	143,591	143,591	143,591	143,591
%  141,346	135,215	118,684	37,872	25,955	17,936
%  98.44\%	94.17\%	82.65\%	26.37\%	18.08\%	12.49\%
		
	 & \textit{minRe} & 10 & 12 & 14  &  16  & 18  & 20  \\
		\textbf{(f) T40I10D100K}	& HUPs &  143,591 &	143,591 &	143,591 &	143,591 &	143,591 &	143,591 \\
		($\delta$: 0.05\%, \textit{minUtil}: 0.2\%)	& RHUPs & 141,346 &	135,215 &	118,684 &	37,872 &	25,955 &	17,936  \\
		& \textit{recencyRatio} & 98.44\% &	94.17\% &	82.65\% &	26.37\% &	18.08\% &	12.49\%  \\	    
		\hline

		\hline 	\hline
	\end{tabular}
\end{table}
%%%%%%%%%%%%%%%%%%%%%%%%%%%%%%%%%%%%%%%%%%%%%%%%%%%%%%

In Table \ref{table:patterns1} and Table \ref{table:patterns2}, it can be observed that for normal, sparse or dense datasets, the number of RHUPs is always smaller than the number of HUPs under different parameter settings. 
The reason is that although numerous HUPs are discovered when only the utility of  patterns is considered, few HUPs are  recent patterns since the recency constraint is ignored in traditional HUPM. For example, on foodmart ($ \delta $ = 0.001\%, $ minRe $ = 1.0), when $ minUtil $ is set to 0.002\%,  492,041 HUPs are found but only 41,539 are RHUPs, and when $ minUtil $ is set to 0.007\%, 49,848 HUPs are found but only 15,238/49,848 = 30.57\% is the interesting RHUPs. This  indicates that many HUPs have low recency values. Hence, those HUP patterns may not be helpful for decision-making (e.g., for a manager or retailer) and may be generally uninteresting for users in real-world applications. As  previously explained, the concept of RHUP patterns was introduced in this paper to find patterns  that are both highly profitable and have  recently appeared in transactions. Thus, the discovered RHUPs can be considered more valuable than the HUPs found by traditional HUPM algorithms for  tasks such as product recommendation and promotion.

It can be observed that the compression achieved by mining RHUPs instead of HUPs, indicated by the \textit{recencyRatio}, is very high when the \textit{minUtil} or $minRe$ thresholds are varied. For instance, on the T10I4D100K dataset with $minRe =  6$,  a low \textit{recencyRatio} of 1.09\% is obtained (as shown in Table \ref{table:patterns2}), which means that numerous redundant and meaningless patterns have been effectively eliminated (outdated patterns). In other words, fewer patterns are found by mining HUPs, but those patterns are up-to-date. As more constraints are applied in the mining process, fewer but more meaningful patterns are discovered. It can also be observed that the \textit{recencyRatio}/\textit{outdatedRatio} produced by the RUP algorithm increases/decreases when \textit{minUtil}/\textit{minRe} is increased/decreased. More specifically, the number of HUPs remains unchanged when \textit{minRe} is increased, while the number of RHUPs decreases. For example, on the chess dataset as shown in Table \ref{table:patterns1}(c), when \textit{minUtil} is increased from 17\% to 22\%, the numbers of patterns generated by RUP-based algorithms decreases from 94,666 to 2,249 patterns, and the \textit{recencyRatio} increases from 47.59\% to 75.75\%. Furthermore, when \textit{minUtil} is fixed and \textit{minRe} is varied from 20 to 120, as shown in Table \ref{table:patterns2}(c), the number of HUPs remains 89,933, and the \textit{recencyRatio} decreases from 99.39\% to 0.50\%. This is reasonable since  traditional HUPM  does not consider the minimum recency threshold, while mining RHUPs is done by considering both the utility and recency constraints. It is thus not surprising that fewer RHUPs are produced for high recency threshold values and a fixed minimum utility threshold. The above observations generally hold when the minimum utility threshold is set to large values. When the \textit{minRe} threshold is increased, fewer RHUPs are produced by the proposed RUP algorithms compared to the number of HUPs discovered by traditional HUPM algorithms. This is because the proposed algorithms discover RHUPs by considering both the recency and utility measures. The FHM algorithm however only considers the utility measure to discover HUPs. From these results, it can be concluded that few interesting patterns satisfy high minimum recency threshold values (RHUPs). Thus, a high \textit{recencyRatio} can be obtained when the user-specified minimum recency threshold is high.

Therefore, the analysis of patterns found by the proposed RUP algorithm and existing HUPM algorithm has shown that the RHUPM task can effectively discover fewer but more useful and up-to-date HUPs than the traditional HUPM framework, which only considers the minimum utility threshold. The proposed algorithm can not only dramatically reduce the number of patterns found, but also make the high-utility pattern mining task more suitable for real-life applications.

\subsection{Runtime Analysis}

In the second experiment,  the performance of three variants of the proposed RUP algorithm was compared in terms of runtime. Three conventional algorithms of HUPM (including FHM, HMiner, and EFIM) were also compared as a benchmark. Parameters were set as in the previous subsection. Results when the minimum utility threshold is varied and the minimum recency threshold is fixed, and when $minRe$ is fixed and \textit{minUtil} is varied, are respectively shown in Fig. \ref{fig:time1} and Fig. \ref{fig:time2}.

\begin{figure*}[!hbtp]	
	\centering
	\includegraphics[trim=60 0 50 0,clip,scale=0.37]{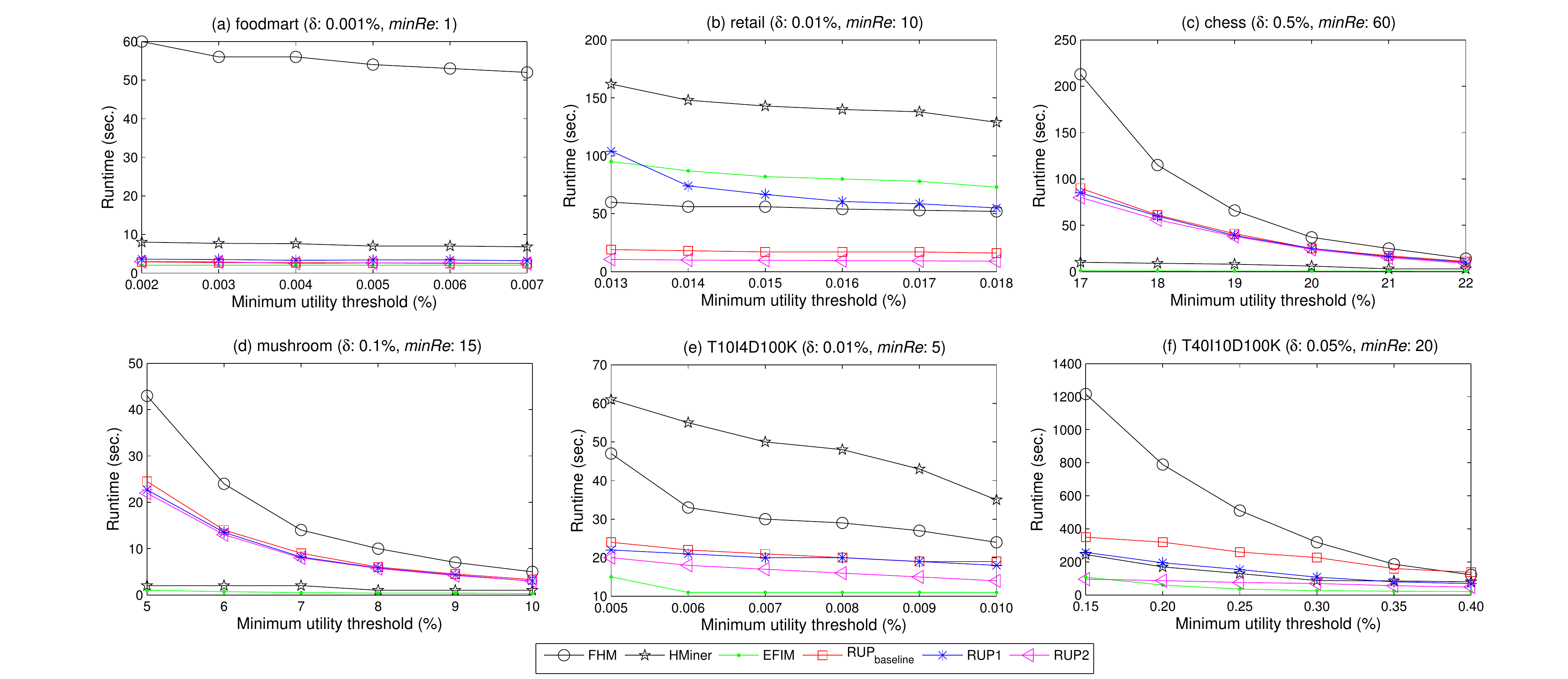}
	\caption{Runtime when \textit{minUtil} is varied and $minRe$ is fixed.}
	\label{fig:time1}
\end{figure*}

In Fig. \ref{fig:time1}, it can be observed that the proposed RUP-tree based algorithms outperforms the FHM algorithm, and that  versions of the RUP algorithms integrating  additional pruning strategies outperform the baseline RUP$\rm _{baseline} $ algorithm when \textit{minUtil} is varied and \textit{minRe} is fixed. In general, RUP2 is about one to two times faster than FHM, but slower than the state-of-the-art EIFM algorithm. This is reasonable since both the utility and recency constraints are taken into account by the RU-list-based RUP algorithm to find RHUPs, while only the utility constraint is considered by the utility-list-based FHM algorithm to find HUPs. When more constraints are applied, the search space can be further reduced and fewer  patterns can be discovered.  Besides, the pruning strategies used in the two improved algorithms are more efficient than those used by the baseline RUP$\rm _{baseline} $ algorithm. The GDC and CDC properties are much more effective at pruning the search space than  RUP$\rm _{baseline} $ and other conventional HUPM algorithms.  When \textit{minUtil} is decreased, it can be further observed that the gap in terms of runtime between FHM and the proposed RUP-based algorithms sharply increases, while the runtimes of the  three proposed algorithms are similar. For example, on the retail dataset and \textit{minUtil} = 0.015\% (as shown in Fig. \ref{fig:time1}(b)), the runtime of HMiner is 143 seconds, but the three other RUP-based algorithms only require 17 seconds, 10 seconds and 9.8 seconds; the gap is quite larger than when \textit{minUtil} is set to 0.013\%. The reason is that  more HUPs can be discovered for lower \textit{minUtil} values, but most of these patterns are not up-to-date, so the runtime of HMiner and FHM is greater than  RUP-based algorithms, which aim at discovering recent patterns.

\begin{figure*}[!hbtp]
	\centering
	\includegraphics[trim=60 0 60 0,clip,scale=0.37]{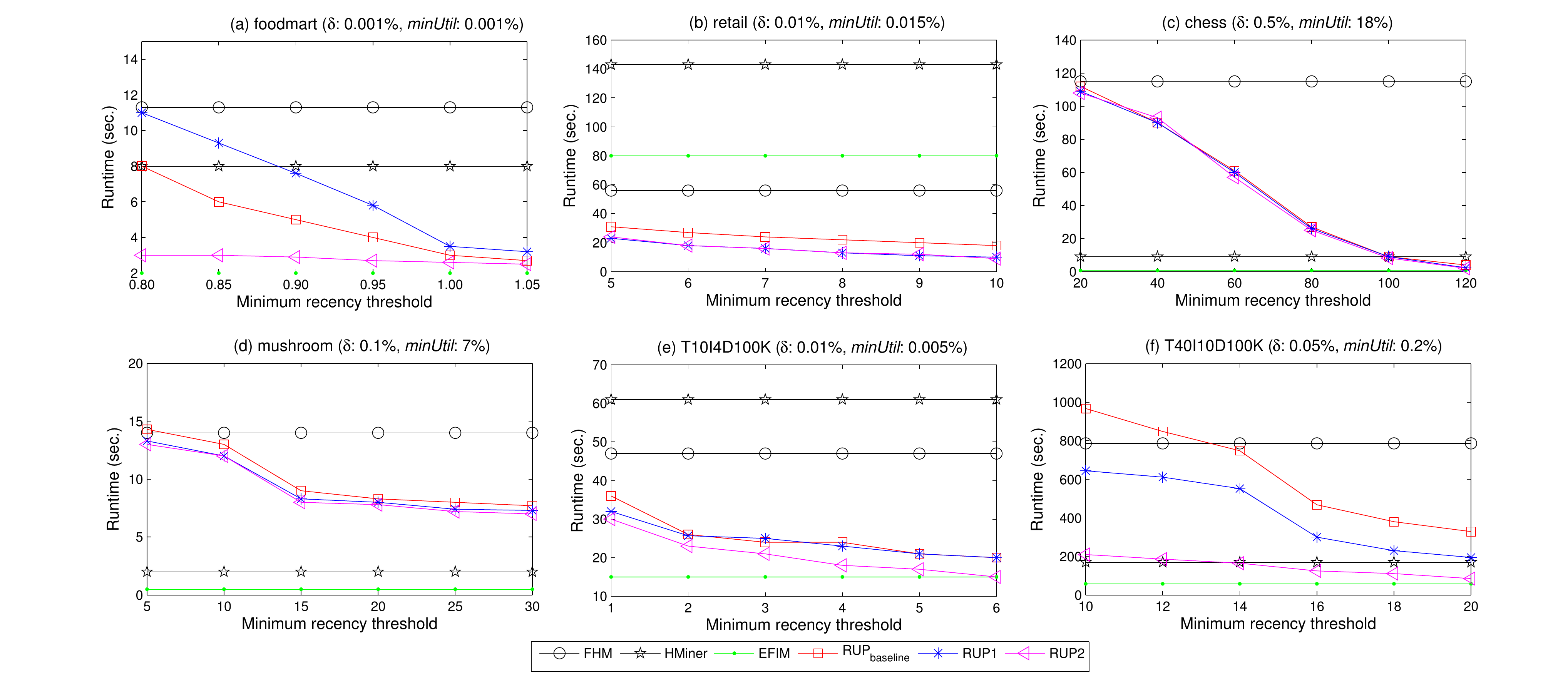}
	\caption{Runtime when \textit{minRe} is varied and \textit{minUtil} is fixed.}
	\label{fig:time2}
\end{figure*}

In Fig. \ref{fig:time2}, it can be seen that both the proposed RUP1 and RUP2 algorithms outperform the RUP$\rm _{baseline} $  and FHM  for various \textit{minRe} values, except for the foodmart dataset. It can be further observed that when \textit{minRe} is increased, the runtime of the two proposed algorithms remain steady, while the runtimes of the RUP1 and RUP2 algorithms sharply decrease,  and the runtime of the FHM algorithm remains the same. For example, consider the results obtained for the dense chess dataset, depicted in Fig. \ref{fig:time2}(c). When the \textit{minRe}  threshold is varied from 20 to 120, the runtime of the FHM algorithm remains 115 seconds, while the runtimes of the RUP$\rm _{baseline} $, RUP1 and RUP2 algorithms decrease from about 112 seconds to 2 seconds.  More specifically, the two improved algorithms are generally faster than the FHM algorithm. This is reasonable for two reasons.  On one hand, when the \textit{minRe}  threshold is set high, the search space can be reduced and fewer RHUPs are produced by the RUP algorithms, while the traditional algorithms (i.e., FHM, HMiner, and EFIM) for mining HUPs is not influenced by the minimum recency threshold. Hence, discovering RHUPs is faster when the \textit{minRe}  threshold is set higher, but the runtime of FHM, HMiner, and EFIM remains the same. On the other hand, based on the proposed RU-tree and RU-list structures, the partial anti-monotonicity provided by the CDC property and the global anti-monotonicity of the GDC property are more effective at pruning the search space,  and thus the runtime of the proposed RUP algorithm is greatly reduced.

In Fig. \ref{fig:time1}(a) and Fig. \ref{fig:time2}(a), it can be seen that for the foodmart dataset, the RUP$\rm _{baseline} $ algorithm performs slightly better than the RUP1 algorithm, and that both of them are slower than the RUP2 algorithm. The reason is that for very sparse datasets such as foodmart, where the average transaction length is 4.4 items, the $TWU$ upper-bound of each transactions is close to the utility of each item in the transaction. Thus, numerous unpromising patterns can be directly pruned by the \textit{TWDC} property in the RUP$\rm _{baseline} $ algorithm. In these cases, the EUCP strategy does not allow pruning much more patterns, but constructing the EUCS is time-consuming. Hence, the enhanced RUP1 algorithm is slightly slower than the baseline algorithm (which does not utilize the EUCP strategy) on the foodmart dataset. For  very dense datasets (e.g., chess and mushroom), the two enhanced algorithms have similar performance, as shown in Fig. \ref{fig:time1}(c), Fig. \ref{fig:time1}(d), Fig. \ref{fig:time2}(c), and Fig. \ref{fig:time2}(d). This is reasonable because all 2-patterns in the \textit{EUCS} have high $TWU$ values for very dense datasets. In this case, the EUCP strategy  cannot efficiently reduce the search space. Therefore, it can be concluded that the EUCP strategy is inefficient for  very sparse or very dense datasets. In summary, the proposed RU-list-based RUP algorithm performs well compared to the FHM algorithm for HUPM, and the improved versions of the RUP algorithm are faster than the baseline algorithm in most cases.

\subsection{Effect of Pruning Strategies}

In a third experiment, we  evaluated  the effectiveness of the novel pruning strategies, integrated in the designed RUP algorithms. Henceforth, the numbers of  nodes visited in the RU-tree by the RUP$\rm _{baseline} $, RUP1, and RUP2 algorithms are respectively denoted as $N_1$, $N_2$, and $N_3$. Experimental results when \textit{minUtil} is varied and \textit{minRe} is fixed are shown in Fig. \ref{fig:nodes1}, and  results when \textit{minRe} is varied and \textit{minUtil} is fixed are shown in Fig. \ref{fig:nodes2}.

\begin{figure*}[!hbtp]
	\centering
	\includegraphics[trim=60 0 60 0,clip,scale=0.38]{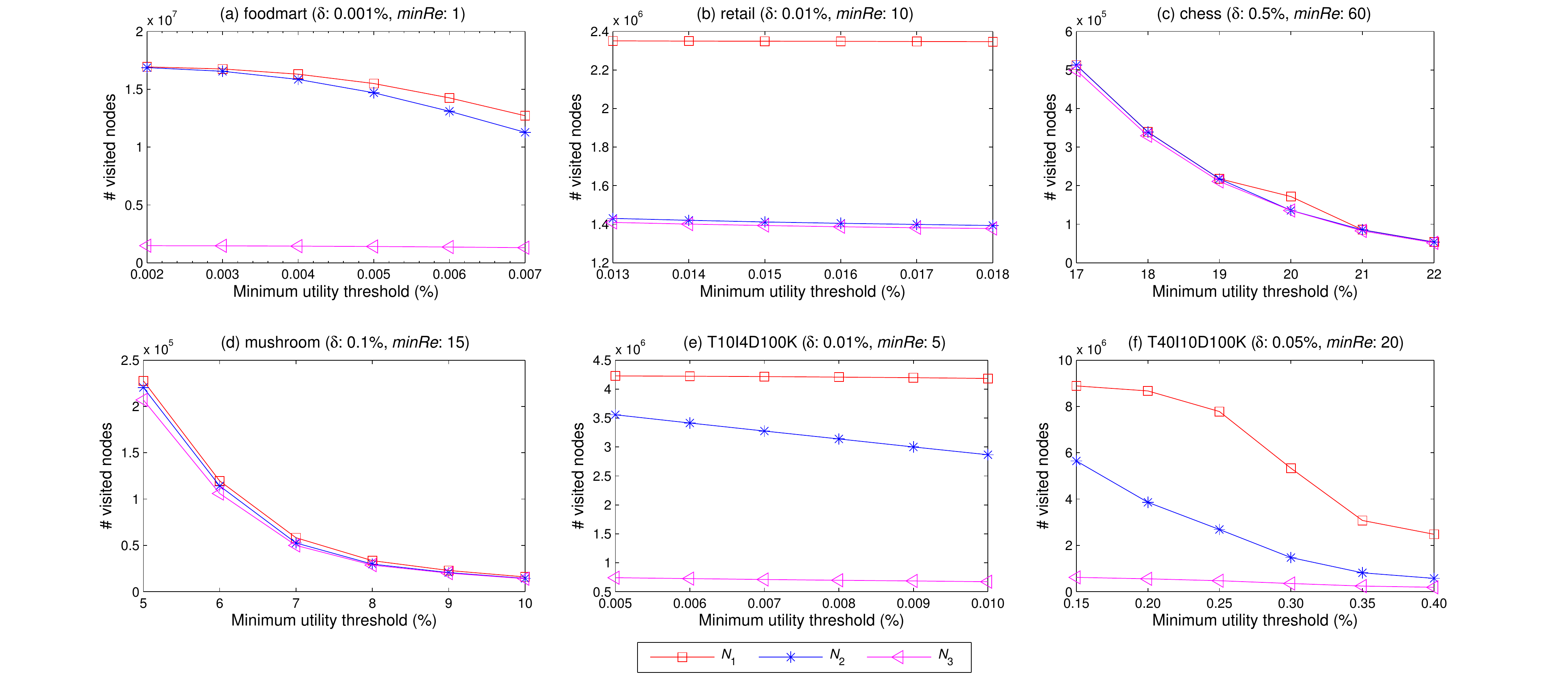}
	\caption{Number of visited nodes when \textit{minUtil} is varied and $minRe$ is fixed.}
	\label{fig:nodes1}
\end{figure*}

\begin{figure*}[!hbtp]	
	\centering
	\includegraphics[trim=60 0 60 0,clip,scale=0.38]{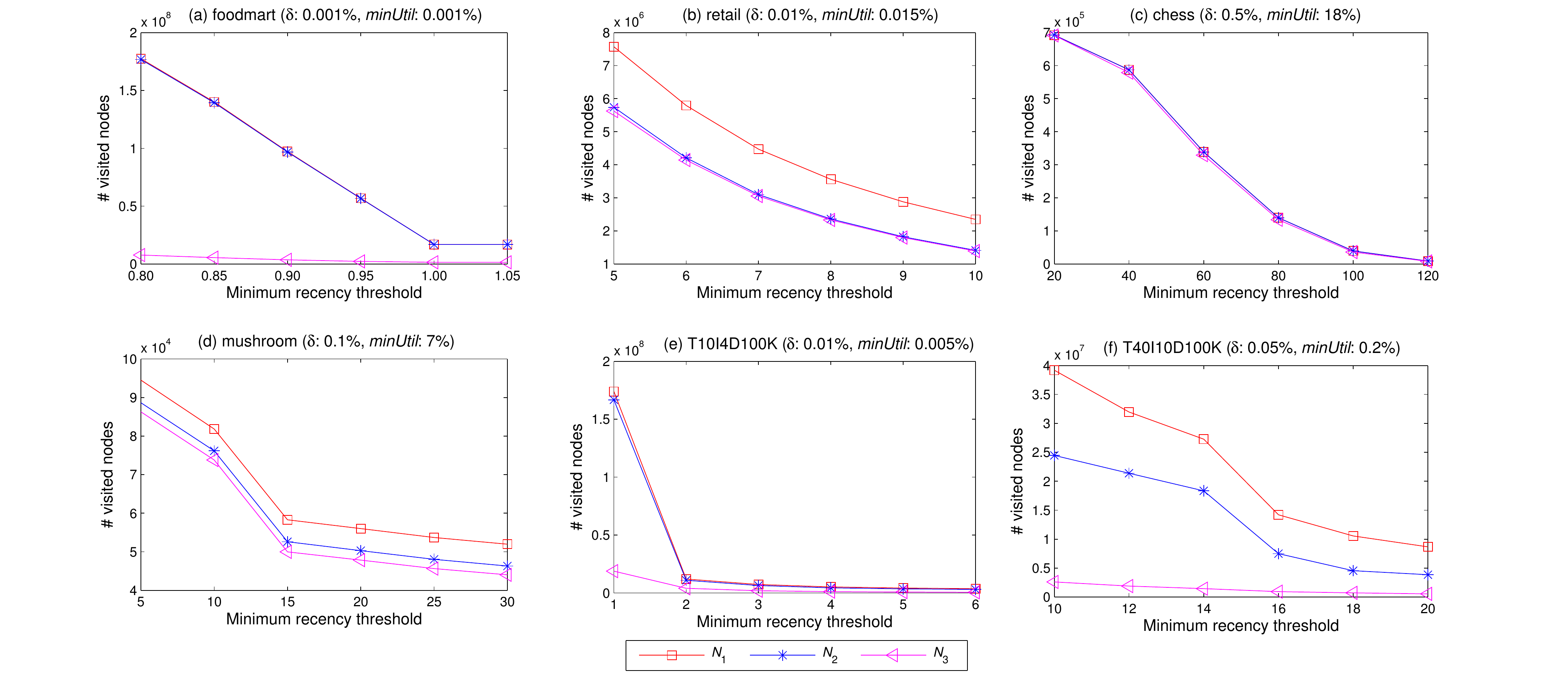}
	\caption{Number of visited nodes when \textit{minRe} is varied and \textit{minUtil} is fixed.}
	\label{fig:nodes2}
\end{figure*}

In Fig. \ref{fig:nodes1}, it can be observed that the various pruning strategies can reduce the search space represented by the RU-tree. In addition, it can also be observed that the pruning strategy 1, which relies on the \textit{TWU} upper-bound and recency values, can  prune several unpromising candidates early. This is very useful since the construction of numerous RU-lists of items and their supersets can be avoided. For example, there are 6470 distinct items in the retail dataset. Thus,  $ 2^{6470} - 1 $ nodes are in the search space for this dataset. When \textit{minRe} = 10 and \textit{minUtil} = 0.015\%, however, the total number of visited nodes by the RUP$\rm _{baseline} $, RUP1, and RUP2 algorithms for the retail dataset are respectively 2350414, 1411859 and 1393223, which is quite less than $ 2^{6470}$ - 1.  Moreover,  pruning strategies 4 and 5 also  play a positive role in pruning unpromising patterns on most datasets when \textit{minUtil} is varied. In both Fig. \ref{fig:nodes1} and Fig. \ref{fig:nodes2}, the relationship $N_1 \geq N_2 \geq N_3 $ always hold. Given these results, the developed pruning strategies used in the proposed RUP algorithms can be considered effective.

It can also be concluded that the pruning strategy 5, integrated in the RUP2 algorithm, can  prune a huge number of unpromising patterns when applied on sparse datasets, as shown on foodmart. Contrarily to strategy 5, for very sparse and dense datasets,  the pruning strategy 4 is no longer effective for pruning unpromising patterns since  those are efficiently filtered using strategies 1 to 3, as shown on foodmart (Fig. \ref{fig:nodes1}(a) and Fig. \ref{fig:nodes2}(a)), chess (Fig. \ref{fig:nodes1}(c) and Fig. \ref{fig:nodes2}(c)), and mushroom (Fig. \ref{fig:nodes1}(d) and Fig. \ref{fig:nodes2}(d)). For these datasets, the number $ N_1 $ is close to $N_2$ but never smaller. Consequently, it is empirically observed that the relationship $N_1 \geq N_2 \geq N_3 $ holds on all datasets and for various parameter values.

\subsection{Scalability}

In this subsection, the scalability of the three proposed  algorithms is compared on the synthetic T10I4N4KD$|\textit{X}|$K dataset. This dataset has been generated using the IBM Quest synthetic data generator \cite{agrawal1994dataset} with various number of transactions $X$ (from 100K to 500K, with an increment of 100K w.r.t. 100,000 transactions in each test). Note that the parameters $\delta$, \textit{minRe} and \textit{minUtil} are set in this experiment to 0.0001, 5, and 0.005\%, respectively. The results in terms of runtime, memory consumption, the number of derived patterns, and the number of visited nodes, are shown in Fig. \ref{fig:scalability}(a) to Fig. \ref{fig:scalability}(d), respectively.

%% % % % % % % % % % % % % % % % % % % % % % % % % % % % %
\begin{figure}[!htbp]
	\centering
	\includegraphics[trim=50 20 40 5,clip,scale=0.48]{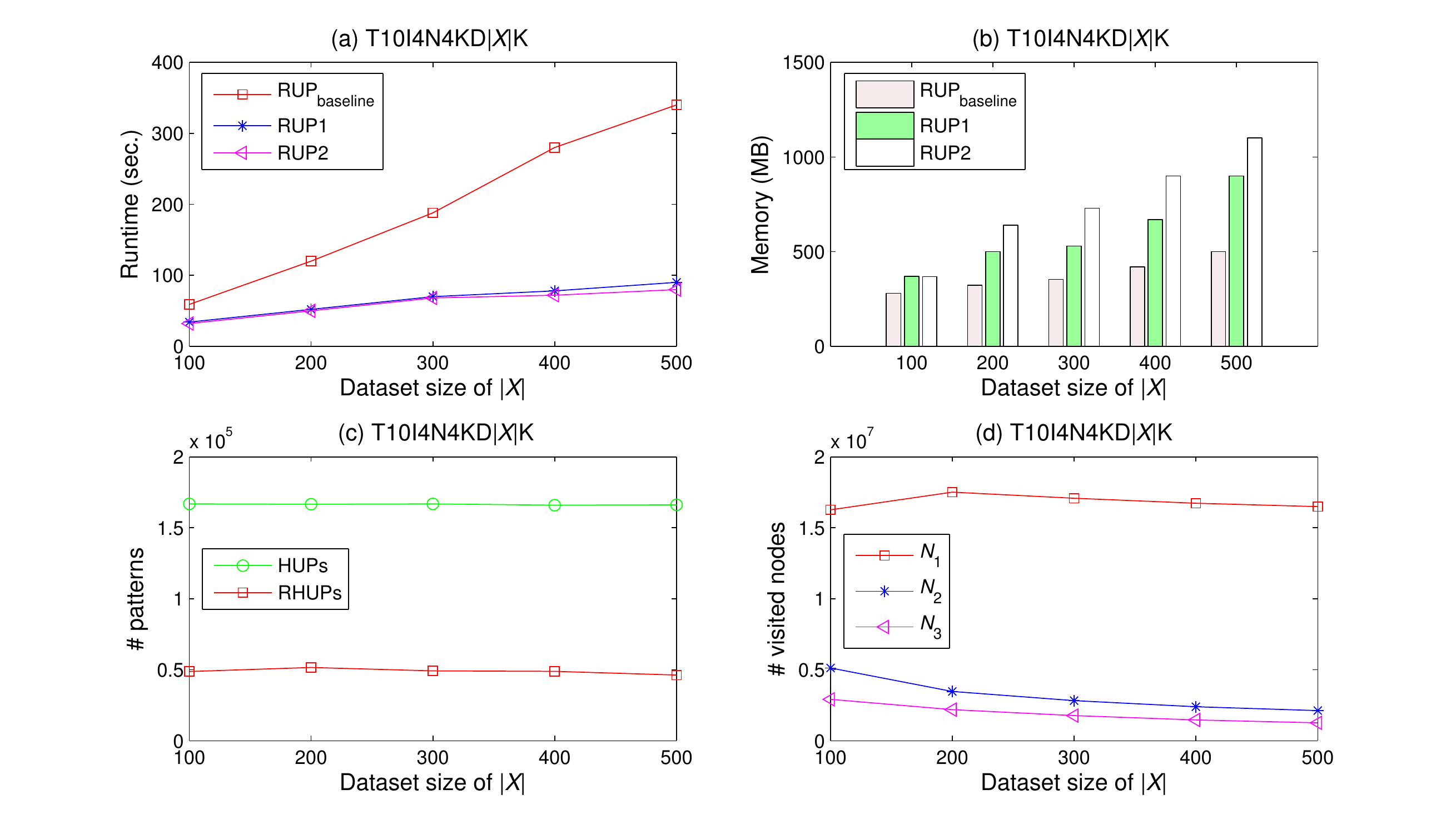}
	\caption{Scalability results.}
	\label{fig:scalability}
\end{figure}
%% % % % % % % % % % % % % % % % % % % % % % % % % % % % %

In Fig. \ref{fig:scalability}, it can be observed that the three proposed  algorithms always have good scalability for various dataset sizes, in terms of runtime, memory consumption, and  number of visited nodes in the RU-tree. The baseline algorithm consumes the smallest amount of memory  but has the longest runtime, and its actual search space in the RU-tree ($N_1$) is also the largest. The improved RUP2 algorithm is always faster than the baseline and RUP1 algorithms but consumes slightly more memory than RUP1. Moreover, it is observed that the runtimes of the three variants of proposed algorithm increase linearly as the dataset size is increased, as shown in Fig. \ref{fig:scalability}(a). Hence, the proposed algorithm scales well for large-scale datasets. Clearly, as the size of the database increases, the overall tree construction and mining time increases. From the above results, it can be concluded that the proposed RUP approach and its improved versions have  acceptable performance for real-world applications.

\section{Conclusions}

Up-to-date knowledge is more interesting and useful than outdated knowledge. How to efficiently and effectively mining of utility-driven trend information is challenge. In this paper, an efficient utility mining algorithm named RUP has been designed to discover recent high-utility patterns (RHUPs) in temporal databases by taking into account both the recency and utility constraints. This addresses an important limitation of traditional HUPM algorithms that is to produce numerous invalid and outdated patterns. To discover RHUPs, this paper has proposed a compact recency-utility tree (RU-tree) structure, used to store the information about patterns that is required to discover RHUPs.  The RUP algorithm performs a depth-first search to explore the conceptual RU-tree and construct the RU-lists.  A substantial experimental evaluation has shown that the proposed RUP algorithm and its improved versions can efficiently identify a set of recent high-utility patterns in time-sensitive databases. The developed RUP-based algorithms perform better than the conventional utility mining algorithms. Moreover, it was observed that  the two improved variants perform better than the baseline one.

%%%%%%%%%%%%%%%%%%%%%%%%%%   Acknowledgment   %%%%%%%%%%%%%%%%%%%%%%%%%
\section{Acknowledgments}

We would like to thank the anonymous reviewers for their detailed comments and constructive suggestions for this paper. This research was partially supported by the Shenzhen Technical Grant No. KQJSCX 20170726103424709 and No. JCYJ 20170307151733005, the China Scholarships Council (CSC), and the UH-nett Vest Grant No. 6000474 which is supported by the Western Norway University of Applied Sciences.
%%%%%%%%%%%%%%%%%%%%%%%%%%   Acknowledgment   %%%%%%%%%%%%%%%%%%%%%%%%%

%%%%%%%%%%%%%%%%%%%%%%%%%%%%%%%%%%%%%%%%%%%%%%%%%%%%%%%%%%%%%%%%%%

% Bibliography
\bibliographystyle{journalformat}
\bibliography{main}

\end{document}